\documentclass[11pt]{article}

\usepackage{xcolor}
\usepackage{mathtools}
\usepackage{amsmath,amsthm,amsfonts,amssymb}
\usepackage{mathrsfs}
\usepackage{booktabs}
\usepackage{csvsimple}
\usepackage{algorithm, algpseudocode}
\usepackage{appendix}
\usepackage{graphicx} 
\usepackage{float} 
\usepackage{subcaption} 
\usepackage{threeparttable}
\usepackage{multirow}
\usepackage{multicol}
\usepackage{bbm}
\usepackage{lipsum}
\usepackage[hidelinks]{hyperref} 
\usepackage{authblk}
\usepackage{enumitem}
\hypersetup{urlcolor=black}
\usepackage[comma]{natbib}
 \bibpunct[, ]{(}{)}{,}{a}{}{,}%

\newtheorem{theorem}{Theorem}
\newtheorem{lemma}{Lemma}
\newtheorem{proposition}{Proposition}
\newtheorem{corollary}{Corollary}
\theoremstyle{definition}

\theoremstyle{remark}
\newtheorem{remark}{Remark}

\newcommand{\ex}[2][]{{\bf\sf E}_{#1}\left[#2\right]}  

\newcommand{\var}[2][]{{\bf\sf Var}_{#1}\left[#2\right]}
\newcommand{\cov}[2]{{\bf\sf Cov}\left[#1,#2\right]}

\newcommand{\prob}{{\bf\sf P}}  
\newcommand{\event}[1]{\left\{ #1 \right\}}
\newcommand{\roundown}[1]{\left\lfloor #1 \right\rfloor}
\newcommand{\absV}[1]{\left| #1 \right|}
\newcommand{\rBrac}[1]{\left( #1 \right)}
\newcommand{\sBrac}[1]{\left[ #1 \right]}
\newcommand{\cBrac}[1]{\left\{ #1 \right\}}
\newcommand{\norm}[2][2]{\left\| #2 \right\|^{#1}}

\newcommand{\reals}{{\mathbb R}}
\newcommand{\naturals}{{\mathbb N}}
\newcommand{\indic}{{\mathbbm 1}}
\newcommand{\field}{{\mathcal F}}

\newcommand{\ts}{^{\bf\sf T}}
\newcommand{\inv}{^{-1}}
\newcommand{\eqdef}{\stackrel{\vartriangle}{=}}
\newcommand{\bigO}[2][2]{\mathcal{O} \left({#2}^{#1}\right)}
\newcommand{\Normal}[2][0]{\mathcal{N}\left(#1,#2\right)}


\allowdisplaybreaks[4]   

\title{Black-box Optimization with Simultaneous Statistical Inference for Optimal Performance}
\date{}

\author[a]{Teng Lian}
\author[a]{Jian-Qiang Hu}
\author[b]{Yuhang Wu}
\author[b]{Zeyu Zheng}

\affil[a]{Department of Management Science, Fudan University, China}
\affil[b]{Department of Industrial Engineering and Operations Research, University of California Berkeley, US}

\begin{document}
\maketitle

\begin{abstract}
Black-box optimization is often encountered for decision-making in complex systems management, where the knowledge of system is limited.  Under these circumstances, it is essential to balance the utilization of new information with computational efficiency.  In practice, decision-makers often face the dual tasks of optimization and statistical inference for the optimal performance, in order to achieve it with a high reliability.  Our goal is to address the dual tasks in an online fashion.  \cite{wu2022} point out that {the sample average of performance estimates generated by the optimization algorithm needs not to admit a central limit theorem (CLT)}.  We propose an algorithm that not only tackles this issue, but also provides an online consistent estimator for the variance of the performance.  Furthermore, we characterize the convergence rate of the coverage probabilities of the asymptotic confidence intervals.
\end{abstract}

\section{Introduction}  \label{sec.intro}
Black-box, according to \cite{audet2017}, is defined as ``\emph{any process that when provided an input, returns an output, but the inner workings of the process are not analytically available}.''  
\cite{golovin2017} states that ``\emph{any sufficiently complex system acts as a black-box when it becomes easier to experiment with than to understand}.''  Hence, black-box optimization serves as a highly valuable tool for complex systems management, as an appropriate mathematical model for the behavior of a complex system may be unknown, imprecise, or computationally expensive to evaluate.  
Besides, quantifying the solution quality (which will be rigorously defined later) is a critical task for reliability assessment, of which a standard (offline) approach is statistical inference based on Monte Carlo simulation.  
Hereafter, black-box optimization with statistical inference for the optimal performance will be referred to as the \emph{dual tasks of optimization and statistical inference}.

In recent years, streaming data processing has gained significant attention from both industrial practitioners and academic scholars, driven by rapid developments in the Internet of Things (IoT) and 5G communication technologies.  Amazon Web Service claims that ``\emph{the nature of enterprise data and the underlying data processing systems have changed significantly}'' and ``\emph{stream processing system is beneficial in most scenarios where new and dynamic data is generated continually}'' (\url{https://aws.amazon.com/what-is/streaming-data/}).  Not coincidentally, researchers in both computer science \citep{ahmad2017unsupervised} and operations research \& management science \citep{iquebal2018change} have increasingly focused on streaming data change points analysis.
Detecting such change points often requires substantial computational resources, which could otherwise be allocated to optimization tasks.
Effectively addressing the dual tasks of optimization and statistical inference within streaming data constitutes a critical challenge.

To provide more motivation of the dual tasks within streaming data contexts, consider the case where the decision-makers repeatedly solve similar problems over $n$ periods, such as portfolio optimization problem on every exchange day and online contents recommendation etc.
In a typical offline setting, given the complexity and uncertainty involved, decision-makers typically formulate various decisions based on their expertise, knowledge, and experience. Subsequently, they evaluate the performance under these decisions to determine which one yields the best results. However, each individual task necessitates a sufficiently large sample to attain a precise and reliable outcome. Consequently, the separation of these dual tasks incurs enormous costs. For example, online platforms (e.g., Airbnb, Amazon, JD, Taobao, etc.) usually design various advertisement recommendation systems and then conduct A/B tests to select the system with the highest Click-Through Rate (CTR), which requires a substantial amount of data. Thereafter, they must repeat these procedures as many times as possible to arrive at an adequately informed decision. Intuitively, the information provided by the preceding tests holds the potential to aid in evaluating the performance of a new system. 
In addition, if there exists some change points during the tests, extra costs occur for change points analysis so as to utilize new information.
This raises the question: why not simultaneously address these dual tasks in an online fashion to balance the utilization of new information with computational efficiency?

Consider a black-box that outputs a realization of $Y(\theta)$ once the decision variables $\theta$ is input, which we will refer to henceforth as the (input) parameter vector. The closed form of the performance $\ex{Y(\theta)}$ (the expected value of $Y(\theta)$) is unknown. Suppose that the parameter vector takes values in a compact set $\Theta \subseteq \reals^p$, where $p$ is a positive integer. Assuming that the realizations of $Y(\theta)$ are independent, a widely-used model of the random output is given by
\begin{equation}
\label{def.Y}
Y(\theta) = \mu(\theta) + \sigma(\theta) \cdot \epsilon_\theta, \quad \theta \in \Theta,
\end{equation}
where $\mu(\theta)=\ex{Y(\theta)}$ represents the expected outcome, $\sigma(\theta)$ the standard deviation, and $\epsilon_\theta$ an independent noise with a mean of 0 and unit variance. It is important to note that, for different choices of $\theta$, the probability distribution of $\epsilon_\theta$ is not necessarily identical. 
The primary objective is to simultaneously address two critical tasks. First and foremost, we aim to perform the task of optimization, such as minimizing the performance:
\begin{equation} \label{obj.opt}
\min_{\theta \in \Theta} \mu(\theta).
\end{equation}
Additionally, we aspire to accomplish the task of statistical inference for $\min_{\theta \in \Theta}\mu(\theta)$ simultaneously. In other words, we seek not only some minimizer $\theta^* \in \mathop{\arg\min}\limits_{\theta \in \Theta} \mu(\theta)$ but also aim to provide the associated confidence interval (CI) of $\mu(\theta^*)$ in an online fashion.

In the stream of simulation optimization literature, various methods have been developed to solve \eqref{obj.opt} when neither performance function nor gradient information are analytically available and only realizations can be observed. 
Fundamental techniques include random search \citep{andradottir2006, andradottir2015}, sample average approximation \citep{kim2015}, stochastic approximation \citep{robbins1951, chau2015}, and surrogate-based methods \citep{hz2021} as known as metamodel-based methods \citep{barton2006, barton2009} or response surface methodology.  Benefiting from the development of computing power, a recent trend is to incorporate some optimization techniques into surrogate-based methods, e.g., STRONG \citep{chang2007,chang2013}, SPAS \citep{fan2018} and SMART \citep{zhang2022smart}.
However, most of these methods do not adapt to the online computing environment.  In addition, as is suggested by \cite{audet2017}, ``if gradient information is available, reliable, and obtainable at reasonable cost, then gradient-based methods should be used.''  Therefore, we are interested in the standard gradient-based method: the Kiefer-Wolfowitz (KW) stochastic approximation (SA) algorithm \citep{kiefer1952, kushner1978, spall1992spsa, lecuyer1998, kushner2003}. Moreover, \cite{lecuyer1998} investigate the convergence rate for mean square error (MSE), and \cite{kushner2003} establish the central limit theorem (CLT) for the parameter vector generated by the SA algorithm.  As for online optimization, \cite{hong2020finite} show that the regret of KWSA algorithm for multi-product dynamic pricing with unknown demand is of order $\sqrt{T}$, which is further extended to SPSA by \cite{yang2024}. Surprisingly, few works address the issue of inference for performance, despite its significant relevance in practical scenarios.

In pursuit of valid statistical inference for $\mu(\theta^*)$, two predominant viewpoints have emerged: Bayesian and frequentist.  Bayesian optimization typically employs a Gaussian process as a prior, with the inference component also referred to as stochastic kriging or Gaussian process regression \citep{sun2014, sun2018, meng2022, wang2023}.  Within the realm of frequentist inference, a naive approach involves conducting stochastic approximation until a reliable estimate of $\theta^*$ is obtained, followed by drawing a sufficiently large sample under this estimate to perform inference for $\mu(\theta^*)$. However, this method is inefficient and lacks adaptability to the online environment.  To our knowledge, actor-critic (AC) algorithms \citep{konda1999, konda2003, bhatnagar2009, sutton2020, zhang2022smart} combined with multi-time-scale SA \citep{borkar1997, borkar2008, bhatnagar1997, bhatnagar2000} present potential methods for addressing both optimization and statistical inference tasks concurrently.  Nevertheless, in a black-box setting, \citet{wu2022} have pointed out that the samples used for stochastic approximation may not be sufficiently representative to provide the best estimator for performance, even when choosing the ``fastest'' time scale for performance estimation.  They propose a four-point method to address this issue, but the construction of CIs in practical use is performed offline fashion and does not adapt to the online environment.

In this paper, we propose another approach to address the issue mentioned above. In the field of machine learning, constant step size (denoted by $\gamma$) is favored due to its exponential convergence in $\mathcal{L}^2$ towards a neighborhood of $\mu(\theta^*)$ bounded by $\bigO[]{\gamma}$.  Remarkably, when applied to inference for $\mu(\theta^*)$, this technique aligns with exponential smoothing \citep{hyndman2008}, which has its origins in the work of \cite{brown1963}.  We amalgamate these concepts to devise an algorithm that employs Simultaneous Perturbation Stochastic Approximation \citep[SPSA,][]{spall1992spsa} with standard decreasing step sizes for optimization and exponential smoothing for statistical inference of performance.  Under appropriate conditions, our algorithm inherits the convergence rate of SPSA for optimization, and furthermore it provides a normalized  (performance) estimator that asymptotically behaves like an Ornstein-Uhlenbeck (OU) process.  Finally, by utilizing the estimation techniques for parameters of OU processes, we construct an online consistent estimator for the variance of the performance within the algorithm for practical use.

The contributions of our work in this paper can be summarized as follows:
\begin{itemize}
    \item In the black-box setting, we provide an online algorithm that simultaneously addresses the dual tasks of optimization and statistical inference for optimal performance.
    \item We establish a CLT for the performance estimator and propose an online consistent estimator for the variance of the performance within the algorithm.
    \item In terms of statistical inference, we characterize the convergence rate of the coverage probabilities of the asymptotic CIs, using it as a measure to validate statistical inference for the optimal performance.
    \item To the SA theory, we extend the multi-time-scale technique to a simple combination of decreasing and constant step sizes.
\end{itemize}

The rest of this paper is organized as follows. Section~\ref{sec.algor} describes the details of our algorithm, offering insight and heuristic arguments.
Theoretical guarantees, specifically convergence rate analysis for the parameter vector and the CLT for the performance estimator of our algorithm, are provided in Section~\ref{sec.analy}.  Section~\ref{sec.compr} compares the coverage probabilities of the asymptotic CIs generated by our algorithm with those of other methods in the literature.
In Section~\ref{sec.exprm}, we present numerical experiments to illustrate and evaluate the algorithm's performance.
Finally, Section~\ref{sec.concl} offers a conclusion and some potential future research directions.

\section{Algorithm} \label{sec.algor}
In this section, we propose an online algorithm that performs the dual tasks of optimization and statistical inference in a black-box setting that produces conditional independent random outputs for any given parameter vector. The term ``online'' means that the algorithm (a) applies the realizations to updating the parameter vector and the estimates for both the performance and its variance, and (b) needs not to store these realizations for statistical inference.

\subsection{Algorithm description}
Suppose that a sample of $\tau$ realizations of $Y(\theta)$ is available after the value of the parameter vector $\theta$ is set, where $\tau$ is a positive integer that will be henceforth referred to as batch size. For optimization, we resort to finite-difference-based gradient estimator to obtain the descent direction, such as SPSA, which needs a predetermined sequence of perturbation sizes $\{c_k\}$. Therefore, our algorithm requires two perturbation samples, which will be formally defined later, to produce a new value for the parameter vector and update the estimates for inference in the meanwhile. The iterations are indexed by $k=0,1,\cdots,n-1$, where $n$ is the number of total iterations performed, hence we need a sample of size $2n\tau$ in total. Conventionally, let $y_{kj}^\pm,~ j=1,2,\cdots,\tau$, denote the perturbation realizations of \eqref{def.Y} in iteration $k$. Throughout this paper, $\|\cdot\|$ means Euclidean norm.

Now we are ready to describe the details of our algorithm. To begin with iteration $k$, we have precedent parameter vector $\theta_k$, performance estimate $\mu_k$ and variance estimate $v_k$. Let $u_k$ be a $\reals^p$-valued random vector uniformly distributed in the unit sphere $S^p=\{u\in\reals^p: \|u\|=1\}$. Then we draw independent realizations $y_{kj}^\pm,~ j=1,2,\cdots,\tau$, of $Y(\theta_k \pm c_k u_k)$, respectively. We average each perturbation sample
\begin{equation*}
y^\pm_k = \frac{1}{\tau}\sum_{j=1}^\tau y_{kj}^\pm,
\end{equation*}
and then estimate the gradient and performance under $\theta_k$ as
\begin{subequations}
\label{eqs.estimator}
\begin{align}
g_k &= \frac{y_k^+ - y_k^-}{2 c_k} u_k, \label{eq.gradient} \\
\bar y_k &= \frac{y_k^+ + y_k^-}{2}. \label{eq.y_bar}
\end{align}
\end{subequations}
At the end of iteration $k$, we update the parameter vector as well as the estimates for both the performance and its variance as follows
\begin{subequations} \label{eqs.updator}
\begin{align}
\theta_{k+1} &= \Pi_\Theta\left[\theta_k - a_k g_k\right], \label{eq.theta} \\
\mu_{k+1} &= \mu_k + \gamma (\bar y_k - \mu_k), \label{eq.mu} \\
v_{k+1} &= v_k + \frac{1}{k+1} \left[(\bar y_k - \mu_k)^2 - v_k\right], \label{eq.var}
\end{align}
\end{subequations}
where $\Pi_\Theta$ is the operator of projection onto $\Theta$ and $\gamma \in (0,1)$. Repeating it $n$ times, we can obtain $\theta_n$ (the estiamte of the minimizer $\theta^*$), $\mu_n$ (the estimate of the performance $\mu(\theta^*)$), and $v_n$ (the estimate of the performance variance), among which the latter two are used to construct the CI. The complete algorithm is presented as follows.

\begin{algorithm}[htbp]
\caption{~SPSA along with Statistical Inference}
\label{alg.spsasi}
\begin{algorithmic}[1]
\Require initial parameter vector $\theta_0 \in \Theta$; initial estimates $\mu_0,~v_0$; step sizes $\{a_k\}$ and $\gamma \in (0,1)$; perturbation sizes $\{c_k\}$; batch size $\tau$; terminal index $n$; confidence level $\alpha$.
\State $k \gets 0$
\While{$k < n$}
\State Draw a $\reals^p$-valued random vector $u_k$ from the unit sphere $S^p$.
\State Sample realizations $y_{kj}^\pm,j=1,2,\cdots,\tau$ of $Y(\theta_k \pm c_k u_k)$, respectively.
\State Average each perturbation sample to obtain $y_{k}^{\pm}$.
\State Compute $g_k \gets \frac{y_k^+ - y_k^-}{2c_k}u_k,~ \bar y_k \gets \frac{y_k^+ + y_k^-}{2}$.
\State Update $\theta_{k+1} \gets \Pi_\Theta[\theta_k - a_k g_k]$, $\mu_{k+1} \gets \mu_k + \gamma (\bar y_k - \mu_k)$, $v_{k+1} \gets v_k + \frac{1}{k+1}\left[(\bar y_k - \mu_k)^2 - v_k\right]$.
\State $k \gets k + 1$.
\EndWhile
\Ensure $\theta_n$ and $\mu_n \pm z_{\alpha/2} \sqrt{\gamma v_n/2}$ ($(1-\alpha)$-CI).
\end{algorithmic}
\end{algorithm}

Under convex assumptions, as will be shown later, Algorithm \ref{alg.spsasi} enjoys the following properties:
\begin{enumerate}[label={(\roman*)}]
\setlength{\topsep}{0pt}
\setlength{\itemsep}{0pt}
\setlength{\parsep}{0pt}
\setlength{\parskip}{0pt}
\item It uses all $2n\tau$ samples to carry out stochastic optimization, and has a fast convergence rate up to $\ex{\norm[]{\theta_n - \theta^*}} = \bigO[-\frac{1}{3}]{n}$.
\item A CLT holds for $\mu_n$ and is almost as good as the CLT for the case in which $\theta^*$ is known and a sufficient large sample of $Y(\theta^*)$ is used to perform inference for $\mu(\theta^*)$.
\item In addition to the estimator \eqref{eq.mu} for $\mu(\theta^*)$, it provides a consistent estimator \eqref{eq.var} for $\sigma^2(\theta^*)/(2\tau)$.
\end{enumerate}

\subsection{On the estimators} \label{subsec.estimators}
In this subsection, we provide intuition and heuristic arguments on the estimators of the performance and the variance of the performance.  Some discussions on the properties of the random search FD-based gradient estimator are attached in the Appendix \ref{sec.A.unbiasG}.

\subsubsection*{Why performance estimator takes constant step size.} 
There exists an implicit trade-off between optimization and statistical inference: the samples used for stochastic approximation may not be sufficient representative to provide the best estimator for the performance, which is crucially needed for statistical inference. In other words, if we use a pathwise average to estimate $\mu(\theta^*)$, the samples under $\theta_k$'s far away from $\theta^*$ may yield significant bias for inference of $\mu(\theta^*)$.

From the previous works \citep[e.g.,][]{lecuyer1998, hu2023}, we know that the variance of $g_k$ generated by \eqref{eq.gradient} is $\bigO[-2]{c_k}$, and if we further take $a_k = A(k+1)^{-\eta}$, $c_k = C(k+1)^{-\nu}$ for some constants $A\geq 1$, $C>0$, $\eta\in(1/2,1]$ and $\nu\in((1-\eta)/2,\eta-1/2)$, then $\ex{\norm{\theta_k - \theta^*}} = \bigO[]{a_k c_k^{-2}}$. It follows from Taylor expansion and convexity that $\left|\mu(\theta_k) - \mu(\theta^*)\right| = \bigO{\norm[]{\theta_k - \theta^*}}$, then
\begin{equation*}
\sqrt{n}\left|\frac{1}{n}\sum_{k=0}^{n-1} \ex{\mu(\theta_k)-\mu(\theta^*)}\right| = \bigO[\frac{1}{2}-\eta+2\nu]{n}.
\end{equation*}
On the other hand, it follows from Taylor expansion and compactness of $\Theta$ that $\ex{\bar y_k - \mu(\theta_k)} = \bigO{c_k}$, presuming $\nabla\mu(\cdot)$ is continuous, then
\begin{equation*}
\sqrt{n}\left|\frac{1}{n}\sum_{k=0}^{n-1}\ex{\bar y_k - \mu(\theta_k)}\right| = \bigO[\frac{1}{2}-2\nu]{n}.
\end{equation*}
We need both $\frac{1}{2}-\eta+2\nu<0$ and $\frac{1}{2}-2\nu<0$ to hold, which yields the contradiction between (A) ``$\eta \leq 1$'' and (B) ``$\eta > 2\nu + 1/2$ and $\nu>1/4 \implies \eta>1$''.  While \cite{wu2022} propose a four-point method sacrificing more sample resources to address this issue, we come up with sufficient small constant step size, which makes the estimate take heavier weight on the samples under $\theta_k$'s closer to $\theta^*$, so that we can avoid the condition (B) which yields the contradiction.

Intuitively, we can imagine the performance estimator on an ``accelerating'' time scale with respect to the parameter vector. In the time scale of performance, the change of the parameter vector is exponentially decreasing so that $\theta_k$ is almost static with respect to $\mu_k$. In other words, the performance estimator would contain less information of those samples that are not sufficient representative but weigh more on the samples with $\theta_k$'s sufficient close to $\theta^*$ as the algorithm progresses.

\subsubsection*{Variance estimator is expected to be consistent.}
Suppose $\theta^*$ is known, and draw $n$ samples $\{y_k\}_{k=0}^{n-1}$ under $\theta^*$, then the (biased) sample variance $\hat\sigma^2_n$ is given by
\begin{equation}
\label{eq.sampleVar}
\hat\sigma^2_n = \frac{1}{n} \sum_{k=0}^{n-1} (\bar y_k - \bar\mu_n)^2,
\end{equation}
where $\bar\mu_n = \frac{1}{n} \sum_{k=0}^{n-1} \bar y_k$. Compared with \eqref{eq.var}, which is equivalent to
\begin{equation*}
v_n = \frac{1}{n} \sum_{k=0}^{n-1} (\bar y_k - \mu_k)^2,
\end{equation*}
in the context of statistics, $\bar\mu_n$ is the sample average while $\mu_n$ is the exponentially weighted moving average (or exponential smoothing). Intuitively, if $\mu_n$ is somehow as good an estimate as $\bar\mu_n$, then $v_n$ may also be as consistent an estimate as $\hat\sigma^2_n$. We will present rigorous analysis on this in the next section.

\section{Convergence Analysis} \label{sec.analy}
Let $(\Omega,\field,\prob)$ be the probability space induced by Algorithm \ref{alg.spsasi}, where $\Omega$ is the collection of all sample paths generated by the algorithm, $\field$ is a $\sigma$-field of subsets of $\Omega$, and $\prob$ is a probability measure on $\field$. Define the $\sigma$-fields $\field_k = \sigma(\theta_0,\mu_0,v_0,\cdots,\theta_k,\mu_k,v_k)$ for $k=0,1,\cdots$, i.e., the information up to $\rBrac{\theta_k,\mu_k,v_k}$. Assume the feasible region $\Theta$ is a convex set characterized by inequality constraints:
\begin{equation*}
\Theta = \left\{\theta \in \reals^p: h_j(\theta) \leq 0,j=1,\cdots,m\right\},
\end{equation*}
where $h_j(\theta) \leq 0,j=1,\cdots,m$ are continuously differentiable convex functions with $\nabla h_j(\theta) \neq 0$ whenever $h_j(\theta)=0$ \citep[cf.][]{kushner2003}. Therefore, the projection operation in \eqref{eq.theta} is equivalent to adding an extra projection term $Z_k$, resulting in the recursion
\begin{equation}  \label{eq.theta*}
\theta_{k+1} = \theta_k - a_k g_k + a_k Z_k,
\end{equation}
where $a_k Z_k = \theta_{k+1} - \theta_k + a_k g_k$ is the real vector with the smallest Euclidean norm needed to take $\theta_k - a_k g_k$ back to the feasible region $\Theta$.  

For notational convenience, let $\ex[k]{\cdot} = \ex{\left.\cdot\right|\field_k}$, and $\mu^*=\mu(\theta^*)$.  Note that $\ex[k]{\cdot} = \ex{\left.\cdot\right|\theta_k}~w.p.1$.  Let us introduce the following assumptions, which are very similar to those in the literature \citep{spall1992spsa, lecuyer1998}. \par \noindent
{\bf Assumptions:}
\begin{enumerate}[label={\bf A\arabic*}]
\setlength{\topsep}{0pt}
\setlength{\itemsep}{0pt}
\setlength{\parsep}{0pt}
\setlength{\parskip}{0pt}
\item \label{A.model} The model \eqref{def.Y} satisfies that for every $\theta \in \Theta$,
\begin{enumerate}[left=0pt, label={(\alph*)}]
    \item \label{A.model.mu_CD} $\mu(\cdot)$ is third order continuously differentiable.
    \item \label{A.model.mu_US} $\mu(\cdot)$ is strongly convex and the unique minimizer $\theta^*$ lies in the interior of $\Theta$.
    \item \label{A.model.sigma} $\sigma(\cdot)$ is non-negative, and second order continuously differentiable.
\end{enumerate}
\item \label{A.u} The random directions $u_k \stackrel{i.i.d.}{\sim} {\textsf{Uniform}}(S^p)$ for all $k$.
\item \label{A.step} The positive sequences $\{a_k\}$ and $\{c_k\}$ satisfy that
\begin{enumerate}[label = (\alph*)]
    \item \label{A.step.1} $\sum_k a_k = \infty,~ \sum_k a_k^2 < \infty$.
    \item \label{A.step.spsa} $c_k \to 0,~ \sum_k a_k c_k^2 < \infty,~ \sum_k a_k^2/c_k^2 < \infty$.
\end{enumerate}
\end{enumerate}

\begin{remark}
\ref{A.model} makes sure that the problem \eqref{obj.opt} is solvable as a convex problem, even though no analytical information is available: \ref{A.model}\ref{A.model.mu_CD} is consistent with the condition used in Lemma 1 of \cite{spall1992spsa}, which requires the expected output to be sufficiently smooth with respect to the parameter vector $\theta$, \ref{A.model}\ref{A.model.mu_US} guarantees the uniqueness of minizer, and \ref{A.model}\ref{A.model.sigma} provides bound for the variance of the random output. \ref{A.step}\ref{A.step.1} is the typical condition for deminishing step size gradient descent algorithms including SA, while \ref{A.u} and \ref{A.step}\ref{A.step.spsa} are conditions for SPSA. 
\end{remark}

Under these assumptions, standard
results on convergence $w.p.1$ and the convergence rate of SPSA are presented in the proposition below, of which the proof is provided in Appendix~\ref{subsec.A.theta}. In the proof, we verify standard conditions for convergence $w.p.1$ \citep{kushner2003, spall1992spsa}, and we use the techniques in \cite{hu2023} to prove the convergence rate of root mean square error (RMSE) in \cite{lecuyer1998}.

\begin{proposition}
\label{thm.theta}
Suppose Assumptions \ref{A.model}-\ref{A.step} hold. Then
\begin{equation*}
\theta_n \longrightarrow \theta^*~ w.p.1,~ n \to \infty.
\end{equation*}
In addition, if $a_k = A(k+1)^{-\eta}$ and $c_k =  C(k+1)^{-\nu}$ with some constants $A\geq 1$, $C>0$, and $\eta,\nu \in (0,1]$ for $k=0,1,\cdots$, then
\begin{equation*}
\ex{\norm[]{\theta_n - \theta^*}} = \bigO[-\eta/2 +\nu]{n} + \bigO[-2\nu]{n}.
\end{equation*}
\end{proposition}

\begin{remark}
The above result indicates that for each $\eta$, the best choice of $\nu = \eta/6$. Then it follows from $\eta \in (2/3, 1]$ that the fastest convergence rate is $\ex{\norm[]{\theta_n - \theta^*}} = \bigO[-\frac{1}{3}]{n}$ as is shown in literature.
\end{remark}

\subsection{Weak convergence of the performance estimator} \label{subsec.weakConvergence}
In this subsection, we show that a CLT holds for the performance estimate generated by Algorithm~\ref{alg.spsasi}. First, we state a preliminary result, which implies that the sequence $\{\mu_k\}$ generated by \eqref{eq.mu} is uniformly $\mathcal{L}^2$-bounded. Define $F_k$ the distribution function of the normalized  estimate  $(\mu_k-\mu^*)/\sqrt{\gamma}$. This result supports the tightness of $\{F_k\}_{k\geq 0}$.  Additionally, it is critical to establishing the consistency of the variance estimator. 

\begin{lemma} \label{lemma.mu_L2}
Suppose Assumptions \ref{A.model}-\ref{A.step} hold. Then the sequences $\{(\bar y_k,\mu_k)\}_{k\geq 0}$ generated by Algorithm~\ref{alg.spsasi} satisfy
\begin{equation*}
\begin{aligned}
\lim_{n\to\infty}\ex{\absV{\bar y_n - \mu^*}^2} =&~ \frac{\sigma^2(\theta^*)}{2\tau}, \\
\limsup_{n\to\infty} \ex{\absV{\mu_n - \mu^*}^2} =&~ \bigO[]{\gamma}.
\end{aligned}
\end{equation*}
\end{lemma}
\noindent
The proof of Lemma \ref{lemma.mu_L2} is given in Appendix~\ref{subsec.A.mu_L2}.  

Define $m(t) \triangleq \roundown{t/\gamma}$ and
\begin{equation}
U^\gamma_k \triangleq \frac{\mu_k - \mu^*}{\sqrt{\gamma}}, \label{def.sde_u}
\end{equation}
then its constant interpolation is given by
\begin{equation*}
U^\gamma(t) = U^\gamma_k,~ t \in \left[k\gamma, k \gamma + \gamma\right).
\end{equation*}

\noindent
We are now ready to present the convergence theorem.
\begin{theorem}
\label{thm.fclt}
Suppose Assumptions \ref{A.model}-\ref{A.step} hold. Then for some $t>0$, as $\gamma \searrow 0$, $\{U^\gamma(t+\cdot), W^\gamma(t+\cdot)\}$ converge weakly to some limit $\rBrac{U(\cdot), W(\cdot)}$ that satisfies the following stochastic differential equation (SDE)
\begin{equation} \label{sde}
dU(t) = -U(t)dt + \frac{\sigma(\theta^*)}{\sqrt{2\tau}}dW(t),
\end{equation}
where $W(t)$ is a standard Wiener process.
\end{theorem}
\noindent
The proof of Theorem \ref{thm.fclt} is given in Appendix \ref{subsec.A.fclt}.

\begin{remark}
This result can be translated into the (weak) convergence of $U^\gamma_k$. Note that the limit SDE \eqref{sde} is an OU process, and it satisfies that 
\begin{equation*}
\ex{U(t)} = e^{-t} U(0),~~~ \cov{U(t)}{U(s)} = \frac{\sigma^2(\theta^*)}{4\tau}(e^{-|t-s|}-e^{-(t+s)}),
\end{equation*}
and hence
\begin{equation*}
{\lim_{\gamma \searrow 0} U^\gamma_{m(t)} \stackrel{d}{=}~} U(t) \sim \Normal[e^{-t} U(0)]{\frac{\sigma^2(\theta^*)}{4\tau}\left(1 - e^{-2t}\right)}.
\end{equation*}
Then, the result of Theorem~\ref{thm.fclt} implies that
\begin{equation*}  
\lim_{t\to\infty} \lim_{\gamma \searrow 0} \frac{\mu_{m(t)} - \mu^*}{\sqrt{\gamma}} \stackrel{d}{=} \Normal{\frac{\sigma^2(\theta^*)}{4\tau}}. 
\end{equation*}

{
In practice, one should choose the sample size $n$ according to $\gamma$.  Since $m(t) = \roundown{t/\gamma}$, we have
\begin{equation*}
\ex{\frac{\mu_n - \mu^*}{\sqrt{\gamma}}} \approx e^{-n\gamma}\frac{\mu_0 - \mu^*}{\sqrt{\gamma}}, ~~~ \var{\frac{\mu_n - \mu^*}{\sqrt{\gamma}}} \approx \frac{\sigma^2(\theta^*)}{4\tau}\rBrac{1 - e^{-2n\gamma}}.
\end{equation*}
Thus, we need $\lim_{\gamma\searrow 0}\left.e^{-n\gamma}\right/\sqrt{\gamma} = 0$, which requires $n > -\frac{1}{2\gamma}\log\gamma$.  To this end, one is supposed to choose $n_\gamma=\roundown{-\frac{c}{2\gamma}\log\gamma}+1$ for some $c>1$ at least.  
}
\end{remark}

\begin{remark}
In a typical multi-time-scale scheme, the constant step size $\gamma$ is replaced by a sequence of decreasing step sizes such as $\{\gamma_k\}$.  For the purpose of estimation, the sequence of step sizes is supposed to satisfy that $a_k = o(\gamma_k)$ in addition to the standard conditions of decreasing step size, so that the parameter vector would be almost still in the time-scale of the performance estimator. Using a constant step size can be interpreted as laying the performance estimator on an ``accelerating'' time-scale so as to avoid the contradiction mentioned in subsection~\ref{subsec.estimators}.
\end{remark}

\subsection{Consistent variance estimator} \label{subsec.consistency}
In this subsection, we propose a consistent estimator of the variance for valid statistical inference.  Specifically, we indirectly estimate the variance via the deviation term $\frac{\sigma^2(\theta^*)}{2\tau}$ of the OU process \eqref{sde}.  A naive estimator of the variance is the sample variance $\frac{1}{n} \sum_{k=1}^n (\bar y_k - \bar\mu_n)^2$, as discussed in Subsection~\ref{subsec.estimators}, but it does not adapt to the online environment.  If $\mu^*$ was known, and we obtain some discrete-time observations $\{U_k\}$ of the OU process \eqref{sde}, of which the time interval is denoted by $\gamma$. Then, according to \citet*{2022WHY}, we have
\begin{equation*}
\frac{1}{n}\sum_{k=0}^{n-1} (e^\gamma U_{k+1}-U_k)^2 \longrightarrow \frac{(e^{2\gamma}-1) \sigma^2(\theta^*)}{4\tau}~w.p.1,~ \forall \gamma>0.
\end{equation*}
Theorem~\ref{thm.fclt} implies that as $\gamma$ approaches 0, there exists some integer $N_\gamma>0$ such that $\{U^\gamma_k\}_{k\geq N_\gamma}$ behaves statistically similar to discrete-time observations of the OU process \eqref{sde}, then we obtain
\begin{equation} \label{eq.sCons}
\lim_{\gamma \searrow 0} \frac{1}{n\gamma }\sum_{k=0}^{n-1} (e^\gamma U^\gamma_{k+1}-U^\gamma_k)^2 \longrightarrow \frac{ \sigma^2(\theta^*)}{2\tau}~w.p.1,~ n\to\infty. 
\end{equation}

In practice, we can only use $\{\bar y_k,\mu_k\}_{k\geq 0}$ to estimate the variance.  The proposed variance estimator has a recursive form as follows.
\begin{equation}
v_{k+1} = v_k + \frac{1}{k+1} \left[r_k (\bar y_k - \mu_k)^2 - v_k\right].  \label{eq.var*}
\end{equation}
Note that \eqref{eq.var} is a special case of \eqref{eq.var*} with $r_k \equiv 1$. The consistency of our variance estimator is established in the theorem below, of which the proof is given in Appendix \ref{subsec.A.consist}.
\begin{theorem} \label{thm.consist}
Suppose Assumptions \ref{A.model}-\ref{A.step} hold. Then for $n > -\frac{1}{2\gamma}\log\gamma$,  as $\gamma \searrow 0$,
\begin{equation*}
v_n \stackrel{\prob}{\longrightarrow} \frac{\sigma^2(\theta^*)}{2\tau}.
\end{equation*}
\end{theorem}
\noindent
It is an immediate corollary that the normalized estimator provided by Algorithm \ref{alg.spsasi} is asymptotically normal distributed.
\begin{corollary}
Suppose Assumptions \ref{A.model}-\ref{A.step} hold. Then for $n > -\frac{1}{2\gamma}\log\gamma$
\begin{equation*}
\frac{\mu_n - \mu^*}{\sqrt{\gamma v_n/2}} \stackrel{d}{\longrightarrow} \Normal{1},~~~ \gamma \searrow 0.
\end{equation*}
\end{corollary}

\begin{proof}{Proof}
It is straightforward based on the continuous mapping theorem and Slutsky's theorem, in combination with Theorems~\ref{thm.fclt} and \ref{thm.consist}, that
\begin{equation*}
\frac{\mu_n - \mu^*}{\sqrt{\gamma v_n/2}} = \sqrt{\frac{\sigma^2(\theta^*)}{2\tau v_n}} \frac{\mu_n - \mu^*}{\sqrt{\gamma \sigma^2(\theta^*)/(4\tau)}} \stackrel{d}{\longrightarrow} \Normal{1},~ n > -\frac{1}{2\gamma}\log\gamma, \gamma \searrow 0. 
\end{equation*}
\end{proof}

\begin{remark}  
The above corollary states that for each $\varepsilon>0$ and each sufficient small $\gamma>0$, there exists an integer $N_\gamma>0$ such that $\forall n \geq N_\gamma$, $\absV{F_n(x) - F(x)} < \varepsilon,~ x\in \mathcal{C}(F_n)$, where $F_n$ denotes the distribution function of the normalized  estimator, which is defined by
\begin{equation*}
    \frac{\mu_n - \mu^*}{\sqrt{\gamma v_n/2}},
\end{equation*}
and $\mathcal{C}(F_n)$ denotes the continuity set of $F_n$.
Then it gives the asymptotic $(1 -\alpha)$-CI constructed by
\begin{equation*}
\mu_n \pm z_{\alpha/2}\sqrt{\gamma v_n/2},
\end{equation*}
where $z_{\alpha/2}$ is the lower $\alpha/2$-quantile of standard normal distribution.
\end{remark}

\section{Comparative Analysis} \label{sec.compr}
In this section, we compare the coverage probabilities of the asymptotic CIs generated by Algorithm~\ref{alg.spsasi} with those generated by other methods.  Note that the coverage probabilities of the asymptotic CIs equal to the probabilities that the normalized estimates fall in the associated intervals of standard normal distribution.  
Formally, consider a general performance estimator in the following form:
\begin{equation*}
\mu_n = \sum_{k=1}^n \gamma_k \prod_{j=k+1}^n (1-\gamma_j) \bar y_k.
\end{equation*}
Here, $\gamma_k$ represents the sequence of performance step sizes without any assumption.
{
Specifically, let $\cBrac{\theta_k}$ be a sequence of parameter vectors generated by some optimization algorithm.  Then, by Taylor expansion and \ref{A.model}\ref{A.model.mu_CD}, for each $k$, there exists some $\tilde\theta_k^\pm$ in the line segments between $\theta_k$ and $\theta_k \pm c_k u_k$, respectively, so that
\begin{equation*}
y_k^\pm = \mu(\theta_k) \pm c_k \nabla\mu(\theta_k)\ts u_k + \frac{c_k^2}{2}u_k\ts \nabla^2\mu(\tilde\theta_k^\pm) u_k + \sigma(\theta_k \pm c_k u_k) \bar\epsilon_k^\pm,
\end{equation*}
thus we can decompose \eqref{eq.y_bar} as
\begin{equation}
\bar y_k = \mu(\theta_k) + \beta_k + \tilde\epsilon_k,
\label{eq.dec_y}
\end{equation}
where
\begin{equation*}
\begin{aligned}
\beta_k &= \frac{c_k^2}{4}u_k\ts \sBrac{\nabla^2\mu(\tilde\theta_k^+) + \nabla^2\mu(\tilde\theta_k^-)} u_k, \\
\tilde\epsilon_k &= \frac{1}{2}\sBrac{\sigma(\theta_k + c_k u_k) \bar\epsilon_k^+ + \sigma(\theta_k - c_k u_k) \bar\epsilon_k^-}.
\end{aligned}
\end{equation*}
} 
Then, we obtain
\begin{equation}  \label{eq.normalized_est}
\mu_n - \mu^* = \sum_{k=1}^n \gamma_k \prod_{j=k+1}^n (1-\gamma_j) \sBrac{\mu(\theta_k)-\mu^* + \beta_k + \tilde\epsilon_k} - \prod_{j=1}^n (1-\gamma_j) \mu^*.
\end{equation}
Let $V_n$ be some non-increasing positive sequence that 
{the normalized estimator
\begin{equation} \label{normalized_estimator}
\frac{\mu_n - \mu^*}{\sqrt{V_n}\sigma(\theta^*)} \stackrel{d}{\longrightarrow} \Normal{1}.
\end{equation}
We are interested in comparing the convergence rates of the probability that it falls in the interval $[-z, z]$ under different choice of $\{\gamma_n\}$ and $\{V_n\}$, where $z=z_{1-\alpha/2}$.}
Before going further, consider the specific forms of the sequences
\begin{equation*}
a_n = An^{-\eta},\quad c_n = Cn^{-\nu},\quad \gamma_n = \gamma n^{-\delta},\quad V_n = Vn^{-\kappa},
\end{equation*}
where $A,C,\gamma,V$ are positive constants.
Then by simple algebra, the second term on the right-hand side of \eqref{eq.normalized_est} is a higher-order infinitesimal term with respect to $\sqrt{V_n}$.

By \eqref{eq.normalized_est}, the normalized estimator can be separated into endogenous bias and exogenous noise, and hence the error of coverage probabilities can be separated in the same way. {Define
\begin{equation*}
    F_n(z) \triangleq \prob\cBrac{\frac{\mu_n - \mu^*}{\sqrt{V_n} \sigma(\theta^*)} \leq z},
\end{equation*}
then,
}
\begin{equation*}
\absV{F_n\rBrac{z}-F_n\rBrac{-z}-1+\alpha} \leq \ex{\absV{E_1}} + \ex{\absV{E_2}},
\end{equation*}
where
\begin{equation*}
\begin{aligned}
E_1 &= F_n\rBrac{z}-F_n\rBrac{-z}-\Phi\rBrac{-R_n+z}+\Phi\rBrac{-R_n-z}, \\
E_2 &= \Phi\rBrac{-R_n+z}-\Phi\rBrac{-R_n-z}-1+\alpha
\end{aligned}
\end{equation*}
with
\begin{equation*}
R_n = \frac{1}{\sqrt{V_n}\sigma(\theta^*)} \sum_{k=1}^n \gamma_k \prod_{j=k+1}^n (1-\gamma_j) \sBrac{\mu(\theta_k)-\mu^* + \beta_k} + o(1),
\end{equation*}
and $\Phi$ denoting the distribution function of standard normal distribution.
Conventionally, denote $\phi$ the density function.
It is worth noting that the exogenous error $E_1$ is independent of $R_n$ and characterizes the convergence speed of the noise term
\begin{equation} \label{term.noise}
S_n = \sum_{k=1}^n \xi_{nk},
\end{equation}
where $\xi_{nk} := \frac{1}{\sqrt{V_n}} \gamma_k \prod_{j=k+1}^n (1-\gamma_j) \tilde\epsilon_k$.  Under proper condition, $E_1$ is bounded by the Berry-Esseen bound for martingale CLT.  
For the endogenous error, by Taylor expansion, it yields
\begin{equation*}
E_2 = -z\phi(z)R_n^2 + o(R_n^2),
\end{equation*}
thus $E_2$ is bounded by $\bigO{R_n}$.  Assumption~\ref{A.model} and the compactness of $\Theta$ imply that for some positive constants $\lambda$ and $L$,
\begin{equation*}
\lambda \norm{\theta_k - \theta^*} \leq \mu(\theta_k)-\mu^* \leq L \norm{\theta_k - \theta^*}.
\end{equation*}
Note that since $\nabla\mu(\theta^*)=0$, $\lambda$ coincides with the infimum of the smallest eigenvalues of $\nabla^2\mu(\theta)$ over $\Theta$, and $L$ the Lipschitz constant of $\nabla\mu(\cdot)$.  Hence, the endogenous error $E_2$ is mainly contributed by the bias term because
\begin{equation} \label{term.bias}
R_n = {\boldsymbol\Theta}\rBrac{\frac{1}{\sqrt{V_n}} \sum_{k=1}^n \gamma_k \prod_{j=k+1}^n (1-\gamma_j) \sBrac{\norm{\theta_k - \theta^*} + \beta_k}},
\end{equation}
where $y_n ={\boldsymbol\Theta}(x_n)$ means $\limsup_{n}\absV{y_n/x_n} < \infty$ and $\liminf_{n}\absV{y_n/x_n} >0$.

\subsection{Quantification of the exogenous error} \label{subsec.e1}
In this subsection, we want to quantify the exogenous error in finite sample.
Firstly, we study the sufficient conditions for the weak convergence of the noise term, because the asymptotic CIs are meaningless without asymptotic normality.  Here we use the following condition in addition, 
\begin{itemize}
\setlength{\labelsep}{1.5em}
\setlength{\itemindent}{1em}
\item[({\bf C})] There exists some constant $d>0$ such that $\sup_{\theta\in\Theta} \ex{\absV{\epsilon_\theta}^{2+d}}<\infty$.
\end{itemize}
\begin{remark}
The effect of condition ({\bf C}) is two-fold.  First, by choosing the decay rate of $V_n$ carefully, condition ({\bf C}) implies the so-called ``conditional Lyapunov condition'', i.e.,
\begin{equation} \label{eq.Lyapunov_condition}
\sum_{k=1}^{n} \ex[k]{\absV{\xi_{nk}}^{2+d}} \stackrel{\prob}{\longrightarrow} 0,~~~n\to\infty,
\end{equation}
which further implies ``conditional Lindeberg condition'': for any $\varepsilon > 0$,
\begin{equation*}
\sum_{k=1}^n \ex[k]{\xi_{nk}^2 \indic\cBrac{\absV{\xi_{nk}} > \varepsilon}} \stackrel{\prob}{\longrightarrow} 0,~~~n\to\infty,
\end{equation*}
where $\indic$ is the indicator function.
Consequently, if there exists some constant $\sigma^2<\infty$ such that
\begin{equation} \label{eq.var_cvg}
\sum_{k=1}^{n} \ex[k]{\xi_{nk}^2} \stackrel{\prob}{\longrightarrow} \sigma^2,~~~n\to\infty,
\end{equation}
then the CLT holds for \eqref{term.noise} \citep[Theorem 4, Section 8, Chapter 7,][]{shiryaev2016}.
Secondly, as far as the authors are aware of, condition ({\bf C}) also implies the sufficient condition for the Berry-Esseen bound of martingale CLT \citep[Theorem 2.1,][]{fan2019}, namely, there exists some $\varepsilon_n \searrow 0$ so that for all $k \leq n$
\begin{equation*}
\ex[k]{\absV{\xi_{nk}}^{2+d}} \leq \varepsilon_n^d \ex[k]{\xi_{nk}^2} ~ w.p.1.
\end{equation*}
It yields
\begin{equation*}
E_1 = \bigO[d]{\varepsilon_n}.
\end{equation*}
\end{remark}
The remained question is what decay rate $V_n$ should take so that ``conditional Lyapunov condition'' is satisfied.  
The following proposition formally states the requirement of $V_n$ for the validity of CLT and the associated Berry-Esseen bound.

\begin{proposition} \label{prop.clt}
Suppose Assumptions~\ref{A.model}-\ref{A.step} and condition {\bf (C)} holds.  
If $V_n =\boldsymbol\Theta(\gamma_n)$, 
then the CLT for \eqref{term.noise} holds, i.e.,
\begin{equation*}
S_n \stackrel{d}{\longrightarrow} \Normal{\frac{\sigma^2(\theta^*)}{2\tau}}, ~~~ n\to\infty.
\end{equation*}
In the meanwhile, it also supports the associated Berry-Esseen bound for the martingale CLT
\begin{equation*}
\sup_{z\in\reals}\absV{\prob\cBrac{\sqrt{2\tau} S_n \leq z\sigma(\theta^*)}-\Phi\rBrac{z}} = \bigO[d/2]{\gamma_n}.
\end{equation*}
\end{proposition}
\noindent
The proof of Proposition \ref{prop.clt} is given in Appendix \ref{subsec.A.clt}.

\begin{remark}
Proposition~\ref{prop.clt} quantifies the exogenous error, and it points out that only $\gamma_n$ and the constant $d$ in condition ({\bf C}) matter.  For Algorithm~\ref*{alg.spsasi}, it also gives a Berry-Esseen bound in terms of $\gamma$.  But we can approximate the best finite sample bound by interchanging $n>\frac{1}{2\gamma}\log\frac{1}{\gamma}$ to $\gamma>\frac{\log n}{n}$.
\end{remark}

In fact, we establish that
\begin{equation*}
\begin{aligned}
\absV{E_1} &\leq 2\sup_{z}\absV{F_n(z)-\Phi(z-R_n)} \\
&= 2\sup_{z}\absV{\prob\cBrac{\sqrt{2\tau}S_n\leq z\sigma(\theta^*)}-\Phi(z)} \\
&\leq \bigO[d/2]{\gamma_n}.
\end{aligned}
\end{equation*}
In next subsection, we can show that $E_1$ is usually dominated by $E_2$.

\subsection{Comparison on the endogenous errors}
Now we are ready to study the decay rates of step sizes for performance estimator given the other step sizes (optimization and perturbation), namely the convergence rates of \eqref{term.bias} for the methods given in Appendix~\ref{subsec.benchmarks}.  Taking expectation on the right-hand side of \eqref{term.bias}, it yields
\begin{equation*}
\begin{aligned}
&~ \ex{\frac{1}{\sqrt{\gamma_n}} \sum_{k=1}^n \gamma_k \prod_{j=k+1}^n (1-\gamma_j) \rBrac{\norm{\theta_k - \theta^*} + \beta_k}} \\
=&~ \frac{1}{\sqrt{\gamma_n}} \sum_{k=1}^n \gamma_k \prod_{j=k+1}^n (1-\gamma_j) \rBrac{\ex{\norm{\theta_k - \theta^*}} + \ex{\beta_k}} \\
\leq&~ \bigO[]{\frac{\ex{\norm{\theta_k - \theta^*}}}{\sqrt{\gamma_n}}} + \bigO[]{\frac{\ex{\beta_n}}{\sqrt{\gamma_n}}},
\end{aligned}
\end{equation*}
where the inequality is due to Lemma~\ref{lemma.hu3}.  First of all, for Algorithm~\ref{alg.spsasi}, it is easy to see that fixed $\gamma_n = \gamma \in (0,1)$ leads to
\begin{equation*}
\ex{\absV{R_n}} = \bigO[-\eta+2\nu]{n} + \bigO[-2\nu]{n}.
\end{equation*}
As is proved in the Section~\ref{subsec.e1}, the bias of the normalized estimator of Algorithm~\ref{alg.spsasi} goes to 0 rapidly.  In exchange, it sacrifices the accuracy of the estimation, i.e., it does not give a performance estimator with decreasing variance, but it provides one with variance of order $\gamma$.

\subsubsection*{Four-point method.}
\citet{wu2022} shows that $\ex{\norm{\theta_k - \theta^*}} = \bigO[-1+2\nu]{n}$ with $\nu\in(1/6,1/4)$, and $\ex{\beta_n} = \bigO[-3\nu]{n}$ in four-point method. With $\gamma_n = 1/n$, it follows from Lemma~\ref{lemma.hu3} that
\begin{equation*}
\ex{\absV{R_n}} = \bigO[-1/2+2\nu]{n} + \bigO[1/2-3\nu]{n},
\end{equation*}
which goes to 0 for $\nu\in(1/6,1/4)$, and the fastest rate $\bigO[-1/10]{n}$ is attained when $\nu=1/5$.  

\subsubsection*{Multi-time-scale method.}
Recall that for ordinary SPSA algorithm, $\ex{\norm{\theta_k - \theta^*}} = \bigO[-\eta+2\nu]{n} + \bigO[-4\nu]{n}$ and $\ex{\beta_n} = \bigO[-2\nu]{n}$.  Using $\gamma_n = \gamma n^{-\delta}$ with $\delta<2\gamma/(2+\gamma)$, then
\begin{equation*}
\ex{\absV{R_n}} = \bigO[\frac{\delta}{2}-\eta + 2\nu]{n} + \bigO[\frac{\delta}{2}-2\nu]{n}.
\end{equation*}
The right-hand side above goes to 0 if both $\frac{1}{2} < \delta < \eta \leq 1$ and $\frac{\delta}{2} < (\eta-2\nu)\wedge (2\nu)$ hold, which is equivalent to $\nu \in \rBrac{\frac{1}{8}, \frac{\eta}{2}-\frac{1}{8}}$ and $\delta < 2(\eta-2\nu)\wedge (2\nu)$.   Here, $x\wedge y=\min(x, y)$.  In this case, the fastest rate $\bigO[\frac{\delta-\eta}{2}]{n}$ can be attained by taking $\nu = \eta/4$ with $\delta \in (1/2,\eta)$.

\subsubsection*{Forward SPSA method.}
Following the same procedures in the proof of Proposition~\ref{thm.theta} (see Appendix~\ref{sec.A.proofs}), one can easily show that it turns out $\ex{\norm{\theta_k - \theta^*}} = \bigO[-\eta+2\nu]{n} + \bigO[-2\nu]{n}$ and $\ex{\beta_n} = 0$ for forward SPSA algorithm. Then, it follows that
\begin{equation*}
\ex{\absV{R_n}} = \bigO[\frac{1}{2} - \eta + 2\nu]{n} + \bigO[\frac{1}{2}-2\nu]{n}
\end{equation*}
with $\gamma_n=1/n$, and it needs not to converge to 0 because $-(\eta - 2\nu) \wedge (2\nu) \geq -\eta/2 \geq -1/2$.

\subsection*{Summary}
We have derived the convergence rates of $\ex{\absV{R_n}}$ for Algorithm~\ref{alg.spsasi} and the benchmarks, respectively.  
Notably, $\ex{\absV{E_1}}$ is typically dominated by $\ex{\absV{E_2}}$.  From $E_2 = -z\phi(z) R_n^2 + o(R_n^2)$, it follows that
\begin{equation*}
\ex{\absV{E_2}} = z\phi(z) \ex{R_n^2} + o\rBrac{\ex{R_n^2}} = \boldsymbol\Omega\rBrac{\ex{\absV{R_n}}^2},
\end{equation*}
which dominates $\ex{\absV{E_1}}$ for the four-point method and the forward SPSA method because $(\eta-2\nu)\wedge(3\nu) \leq 3\eta/5 < 3/4$ implies $-1/2<1-2(\eta-2\nu)\wedge(3\nu)$.
The only exception occurs when $\delta \leq \frac{4}{3} (\eta-2\nu)\wedge (2\nu)$ is applied to the multi-time-scale method. However, this is in contradiction with the condition $\nu=\eta/6$, which corresponds to the fastest convergence rate of the parameter vector.
To this extent, the convergence rates of $\ex{\absV{R_n}}$ depict how fast the coverage probabilities approach the confidence level when the normalized noise term approximates the standard normal distribution.  

The bounds of $E_1$ and $E_2$ of our algorithm and the benchmarks are summarized in Table~\ref{table.orders}.  
Notably, for a fixed sample size $n$, the fastest convergence rate of $\ex{\absV{E_1}}$ of our algorithm is attained when $\gamma = \boldsymbol{\Theta}(n\inv\log n)$, resulting in a convergence rate negligibly slower than those of the benchmarks.  The results highlight the efficiency and competitiveness of our algorithm in terms of convergence rates of the coverage probabilities.


\begin{table}[H]
\centering
\caption{Bounds of exogenous ($E_1$) and endogenous ($E_2$) errors of the algorithms. Under Assumption~\ref{A.model}, and condition ({\bf C}) with $d=1$.}
\label{table.orders}
\scalebox{0.85}{
\begin{threeparttable}
\begin{tabular}{ccccc}
\toprule
&\bfseries Algorithm~\ref{alg.spsasi} &\bfseries Four-point &\bfseries Multi-time-scale &\bfseries Forward SPSA \\ \hline
\bfseries $E_1$ & $\bigO[1/2]{\gamma}$ & $\bigO[-1/2]{n}$ & $\bigO[-\delta/2]{n}$ & $\bigO[-1/2]{n}$ \\
\bfseries $E_2$ & $\bigO[-2(\eta-2\nu)\wedge (2\nu)]{n}$ & $\bigO[1-2(\eta-2\nu)\wedge (3\nu)]{n}$ & $\bigO[\delta-2(\eta-2\nu)\wedge (2\nu)]{n}$ & $\bigO[1-2(\eta-2\nu)\wedge (2\nu)]{n}$ \\ 
\bottomrule
\end{tabular}
\begin{tablenotes}[flushleft]
\footnotesize
\item [*] {Note: $\gamma$ is the constant step size used to update performance estimates in our algorithm, while $\delta$ is the decay rate of the step sizes for this purpose in the multi-time-scale method. For the bounds of $E_2$, $\eta$ and $\nu$ are the decay rates of the step size and the perturbation size for optimization, respectively.}
\end{tablenotes}
\end{threeparttable}
}
\end{table}

\section{Simulation Experiments} \label{sec.exprm}
In this section, we demonstrate a 2-dimensional simulation experiment on a set of synthetic examples to compare the performance of Algorithm~\ref{alg.spsasi} with that of the benchmarks in the literature.  Let's begin with the detailed settings.  Consider the following cases in which the output given the input value $\theta$ has different distribution but shares the common mean function $\mu(\theta)$:
\begin{enumerate}[label={Case~\arabic*}, leftmargin=4em]%
\item \label{case.bernoulli} $Y(\theta) \sim {\textsf{Bernoulli}}\rBrac{\mu(\theta)}$, then $\sigma(\theta) = \sqrt{\mu(\theta)\sBrac{1-\mu(\theta)}}$;
\item \label{case.normal} $Y(\theta) \sim \Normal[\mu(\theta)]{\sigma^2(\theta)}$ with $\sigma(\theta)=1.5\sin\rBrac{2\pi\absV{\theta}} + 2.5$;
\item \label{case.gamma} $Y(\theta) \sim {\textsf{Gamma}}\rBrac{\alpha, \mu(\theta)/\alpha}$ with $\alpha=4$, then $\sigma(\theta) = \left.\mu(\theta)\right/ \sqrt{\alpha}$;
\item \label{case.pareto} $Y(\theta) \sim {\textsf{Pareto}}\rBrac{\alpha, (\alpha-1)\mu(\theta)/\alpha}$ with $\alpha=3$, then $\sigma(\theta) = \left. \mu(\theta)\right/ \sqrt{\alpha\rBrac{\alpha-2}}$;
\item \label{case.lognorm} $Y(\theta) \sim {\textsf{Lognormal}}\rBrac{\log\mu(\theta)-s^2/2, s^2}$ with $s=1$, then $\sigma(\theta) = \mu(\theta) \sqrt{e^{s^2}-1}$;
\end{enumerate}
We test the algorithms in the cases above with various functions, which can be categorized into one-dimensional functions and multi-dimensional functions.  We only demonstrate a 2-dimensional function here, the other experiments can be found in Appendix \ref{sec.A.exper}.  
All computations are performed using Python 3.8.11 on a Windows server which has two Intel Xeon Silver 4215R CPUs (with 32 cores in total, each running at 3.20 GHz) and 256 GB of memory.  The codes are available on Github: \url{https://github.com/HenryLean/BBO_Inference}.

Consider a 2-dimensional function
\begin{equation*}
f(\theta) = \frac{1}{2} \theta\ts {\bf A} \theta - {\bf b}\ts \theta + 1,~~~ \theta \in [-2,2]^2,
\end{equation*}
where
\begin{equation*}
{\bf A} = \rBrac{\begin{matrix}
    1.04 & -0.2 \\ -0.2 & 1 
\end{matrix}}, ~~~
{\bf b} = \rBrac{\begin{matrix}
    -1 \\ 0.5
\end{matrix}}.
\end{equation*}
Let $\mu(\theta)=\left.1\right/ \sBrac{1+\exp\cBrac{2-f(\theta)}}$ for \ref{case.bernoulli} and $\mu(\theta) = f(\theta)$ for the others.  We set $\tau=20,~ a_k=30/\rBrac{k+1},~ c_k=1/\rBrac{k+1}^{1/5},~ n=100,000$, and initialize $\theta_0$ uniformly distributed in $[-2,2]^2$.  
The numerical results (averaged over 300 replications) obtained in the test cases are presented in Table~\ref{table.rmses.TwoDim}. For different algorithms in each case, it presents the sample mean of the root mean square error (RMSE) of the parameter vector and that of the {optimality gap} $\absV{\mu(\theta_n)-\mu^*}$ in sub-tables~\ref{subtab.rmse_TwoDim} and \ref{subtab.diff_TwoDim}, respectively. The associated standard deviations are provided in parentheses.  
Additionally, Figure~\ref{fig.TwoDim.rmse} shows the convergence of RMSE with the shadow regions representing 95\% CIs of the Monte Carlo results.
On the performance of optimization, the four-point method shows no significant difference compared to SPSA in terms of RMSE of the parameter vector.  However, forward SPSA results in nearly twice the RMSE of SPSA, even though they share the same order of magnitude.

\begin{table}[htbp]
\centering

\begin{subtable}[h]{0.8\textwidth}
    \centering
    \subcaption{RMSE of the parameter vector} \label{subtab.rmse_TwoDim}
    \csvreader[tabular=cccc, table head= \toprule & \bfseries SPSA$^*$ & \bfseries Four-point & \bfseries Forward SPSA \\ \hline, table foot=\hline]%
    {supplementaries/figures/RMSEs/TwoDim_parameter.csv}{1=\distr, 2=\spsa, 3=\fsp, 4=\fourp}%
    {\distr & \csvcolii & \fourp & \csvcoliii}
\end{subtable}



\scalebox{1}{
\begin{threeparttable}
\begin{subtable}[h]{0.8\textwidth}
    \centering
    \caption{Optimality gap}
    \label{subtab.diff_TwoDim}
    \csvreader[tabular=cccc, table head= \toprule & \bfseries SPSA & \bfseries Four-point & \bfseries Forward SPSA \\ \hline, table foot=\bottomrule]%
    {supplementaries/figures/RMSEs/TwoDim_performance.csv}{1=\distr, 2=\spsa, 3=\fsp, 4=\fourp}{\distr & \csvcolii & \fourp & \csvcoliii}%
\end{subtable}

\begin{tablenotes}
\footnotesize
\item[*] {Our algorithm employs SPSA for optimization, as does the multi-time-scale method.}
\end{tablenotes}
\end{threeparttable}
}

\caption{Performance for different cases, based on 300 independent replications.} 
\label{table.rmses.TwoDim}
\end{table}

\begin{figure}[htbp]
\centering
\includegraphics[width=.75\textwidth]{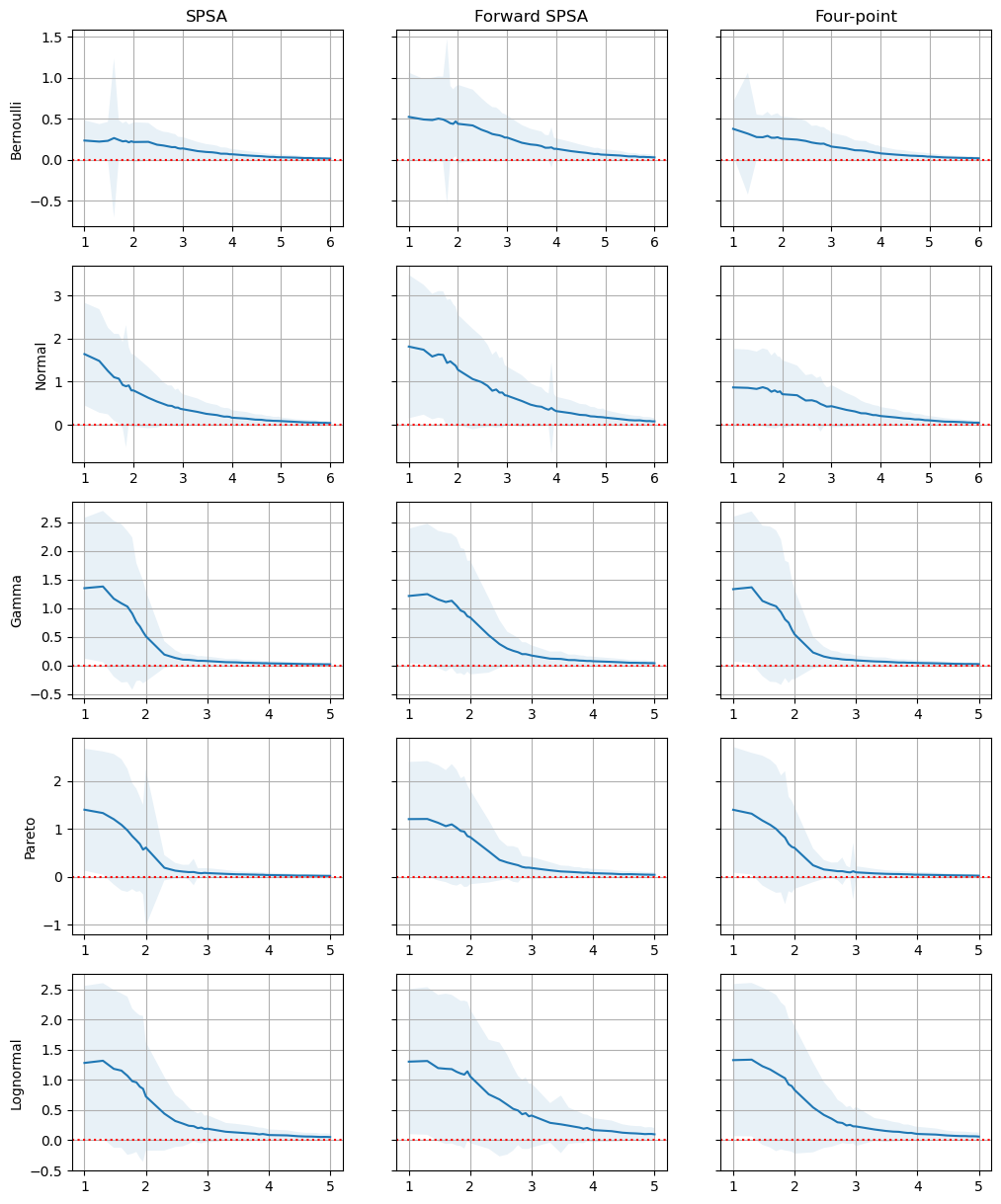}
\caption{RMSEs of input parameters v.s. $\log_{10}\rBrac{\text{\#~iterations}}$ for different cases. The shadow regions are 95\% CIs of the Monte Carlo results.}
\label{fig.TwoDim.rmse}
\end{figure}

Table~\ref{table.tstats.TwoDim} presents the selected statistics of the normalized estimators, and Figures~\ref{fig.TwoDim.histograms} illustrates the corresponding histograms. When considering these results alongside Figure~\ref{fig.TwoDim.coverages}, it is evident that, in the 2-dimensional cases, our algorithm consitently outperforms the benchmarks in terms of statistical inference.
\begin{table}[htbp]
\caption{Selected sample means and standard deviations of the normalized estimators on the 2-d test function, based on 300 independent replications.} \label{table.tstats.TwoDim}
\centering
\begin{filecontents*}{supplementaries/figures/Performances/TwoDim_tStats.csv}
,Ordinary SPSA,Multi-time-scale SPSA,Our algorithm,Forward SPSA,Four-point,Multi-time-scale four-point
Bernoulli,42.3359 (1.0986),3.9614 (1.1061),0.2354 (1.4679),3.8430 (1.6334),0.8342 (1.5095),3.6311 (1.1028)
Normal,47.0423 (1.2261),4.1854 (1.2077),0.3664 (1.6245),26.8711 (1.8758),11.0559 (1.6717),3.9177 (1.2460)
Gamma,163.6675 (1.2673),20.2591 (1.1184),2.5997 (1.5771),32.3414 (1.9295),15.7778 (1.9388),35.4770 (1.1036)
Pareto,144.2563 (2.1835),17.9009 (1.2229),2.2499 (1.5763),33.0564 (1.9062),16.1555 (1.7797),31.3051 (1.8308)
Lognormal,73.9790 (1.1615),8.7991 (1.1104),1.2317 (1.5808),52.4171 (2.3235),23.0410 (1.7573),14.9987 (1.1085)
\end{filecontents*}
\csvreader[tabular=ccccc, table head= \toprule & \bfseries Ordinary SPSA & \bfseries Our algorithm & \bfseries Forward SPSA & \bfseries Four-point \\ \hline, table foot=\bottomrule]{supplementaries/figures/Performances/TwoDim_tStats.csv}{}{\csvcoli & \csvcolii & \csvcoliv & \csvcolv & \csvcolvi}
\end{table}

\begin{figure}[htbp]
\centering
\includegraphics[width=.75\textwidth]{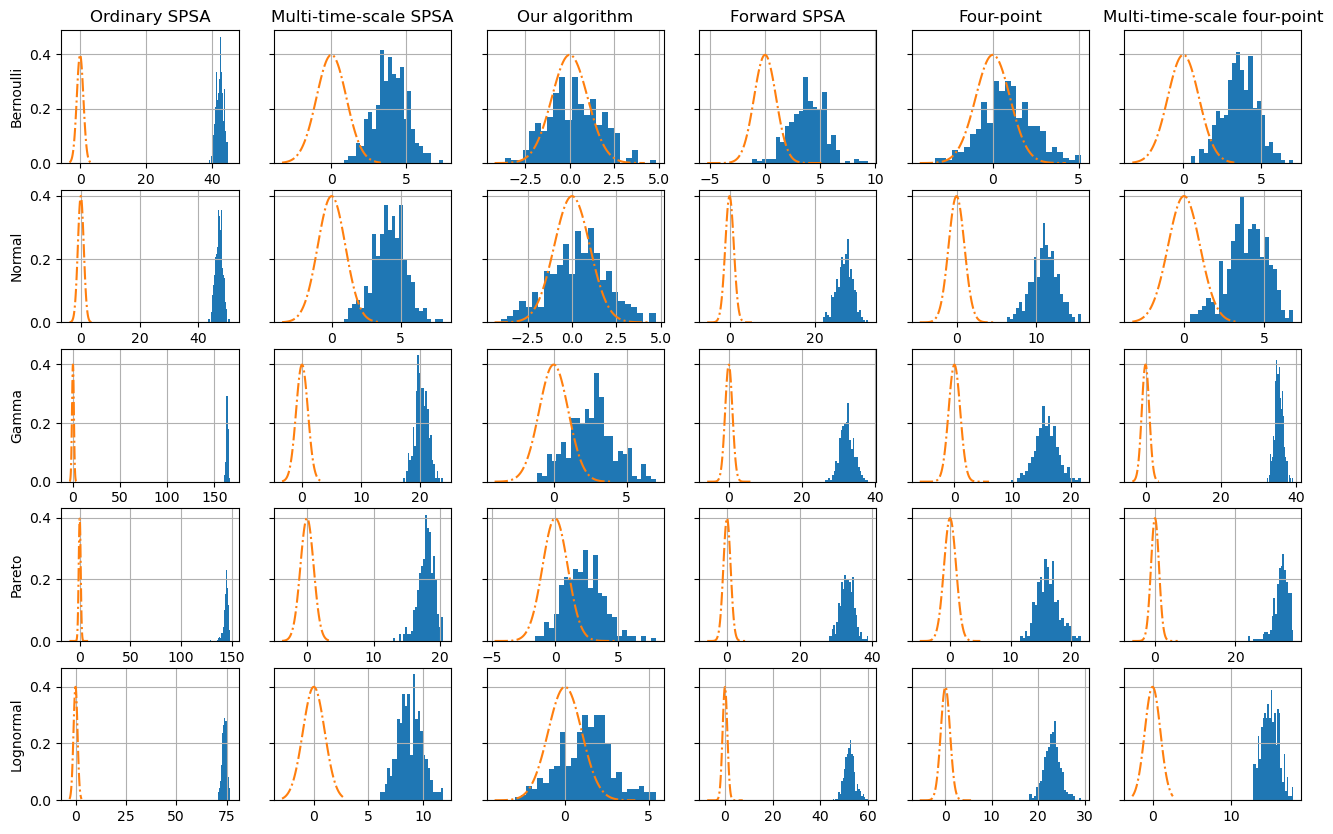}
\caption{Histograms of the normalized estimators on the 2-d test function. The dashed curve is the probability density function of standard normal distribution.}
\label{fig.TwoDim.histograms}
\end{figure}

\begin{figure}[htbp]
\centering
\includegraphics[width=.75\textwidth]{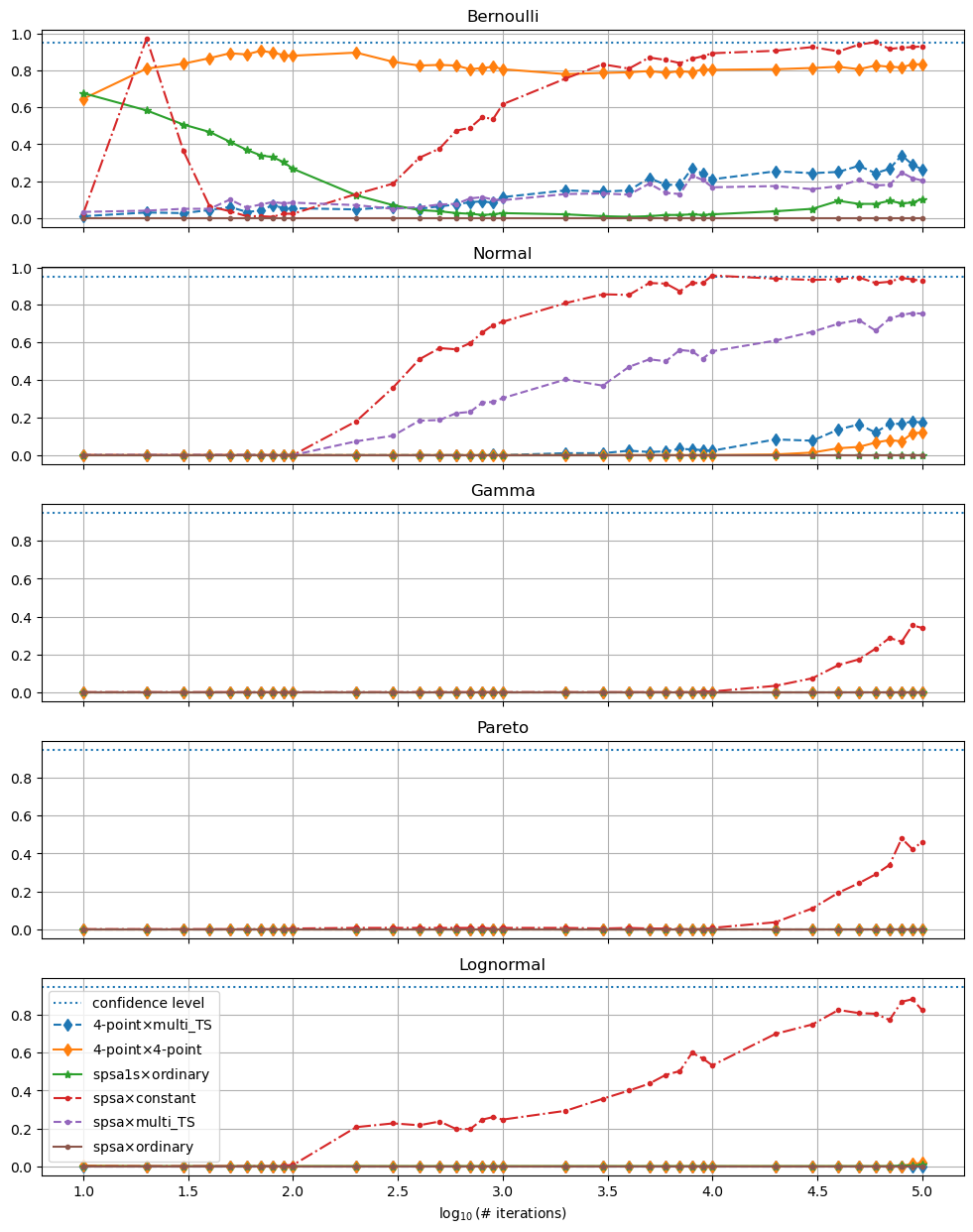}
\caption{Coverage probabilities of the asymptotic CIs on the 2-d test function. The horizontal dashed line with no markers indicates the confidence level.}
\label{fig.TwoDim.coverages}
\end{figure}

In summary, our algorithm outperforms the benchmarks in terms of statistical inference, in the sense that the proposed normalized estimator converges \emph{more rapidly} to the standard normal distribution.  This improvement involves a trade-off between optimal performance estimation and coverage rate, controlled by a user-specified hyperparameter $\gamma$ that linearly adjusts the variance of performance estimator.  By tuning $\gamma$, one can balance precision and reliability, tailoring the algorithm to specific needs.  It offers a significant advantage over the benchmarks, and the numerical results confirm that our algorithm provides a more robust and efficient solution to the dual tasks of optimization and statistical inference.

\section{Conclusion} \label{sec.concl}
To address the dual tasks of optimization and statistical inference simultaneously in a black-box setting, we propose an algorithm based on SPSA for optimization and SA with constant step sizes for statistical inference.  It updates the performance estimator on an ``accelerating'' time scale to prevent it from the contradiction between the CLTs of $\mu(\theta_k) - \mu^*$ and $\mu_k - \mu(\theta_k)$.  In addition, we provide an online consistent estimator for the performance variance, which can be used in constructing the asymptotic CIs.  This paper extends multi-time-scale SA to the combination of decreasing and constant step sizes, where the variable on the ``accelerating'' time scale does not influence the variable on the slower time scale in return.  Furthermore, we characterize the convergence rate of the coverage probabilities of the asymptotic CIs, and show the reason why our algorithm outperforms the benchmarks.

Here we list some potential extensions. It is worth (i) considering more general stochastic black-box models (e.g., hidden Markov model) in the black-box setting; (ii) investigating other optimization algorithms with simultaneous statistical inference when the objective function is non-convex; and (iii) developing theory of multi-time-scale SA or analyzing AC algorithms using the combination of decreasing and constant step sizes.

\bibliographystyle{abbrvnat}
\bibliography{refs}

\begin{appendices}
\newpage
\setcounter{page}{1}
\renewcommand{\thepage}{a-\arabic{page}}

\section{Asymptotic unbiased gradient estimator}  \label{sec.A.unbiasG}
\setcounter{table}{0} 
\setcounter{figure}{0} 
\setcounter{equation}{0} 
\setcounter{theorem}{0}
\renewcommand{\thetable}{\thesection-\arabic{table}}
\renewcommand{\thefigure}{\thesection-\arabic{figure}}
\renewcommand{\theequation}{\thesection-\arabic{equation}}
\renewcommand{\thetheorem}{\thesection.\arabic{theorem}}
As a standard method, under convex assumptions, stochastic approximation has well-developed theory \citep{kiefer1952, kushner2003} guaranteeing that $\theta_n$ is as closed to $\theta^*$ as $n$ is sufficient large. We now present the asymptotic unbiased gradient estimator. Using \eqref{def.Y} we can express the realization as
\begin{equation*}
y_k^\pm = \mu(\theta_k \pm c_k u_k) + \sigma(\theta_k \pm c_k u_k) \bar \epsilon_k^\pm,
\end{equation*}
where $\bar\epsilon_k^\pm$ represents the sample mean of realizations of independent noises for each perturbation sample in iteration $k$. 
By Taylor expansion and \ref{A.model}\ref{A.model.mu_CD}, 
we have
\begin{equation*}
\mu(\theta_k \pm c_k u_k) = \mu(\theta_k) \pm c_k u_k\ts \nabla \mu(\theta_k) + \frac{c_k^2}{2} u_k\ts \nabla^2 \mu(\theta_k) u_k \pm \frac{c_k^3}{6} 
\sum_{i,j,l=1}^p \nabla^3\mu(\tilde\theta_k^\pm)_{ijl} u_{k,i} u_{k,j} u_{k,l},
\end{equation*}
for some $\tilde\theta_k^\pm$ in the line segment between $\theta_k$ and $\theta_k \pm c_k u_k$, respectively, where,
\begin{equation*}
\nabla^3\mu(\theta)_{ijl} = \frac{\partial^3 \mu(\theta)}{\partial\theta_i \partial\theta_j \partial\theta_l},~ i,j,l=1,\cdots,p,
\end{equation*}
with $\theta_i$ representing the $i$th element of $\theta$.  Denote
\begin{equation*} 
\begin{aligned}
\tilde b_k &= \frac{c_k^2}{12} \sum_{i,j,l=1}^p \left[ \nabla^3\mu(\tilde\theta_k^+)_{ijl} - \nabla^3\mu(\tilde\theta_k^-)_{ijl}\right] u_{k,i} u_{k,j} u_{k,l}, \\
\epsilon_k &= \frac{\sigma(\theta_k + c_k u_k)\bar\epsilon_k^+ - \sigma(\theta_k - c_k u_k)\bar\epsilon_k^-}{2 c_k} u_k,
\end{aligned}
\end{equation*}
then, similar to \cite{kushner1978, kushner2003, borkar2008}, we can decompose the gradient estimator \eqref{eq.gradient} as
\begin{equation}
g_k = u_k u_k\ts \nabla\mu(\theta_k) + b_k + \epsilon_k, \label{eq.decomp_g}
\end{equation}
where $b_k = u_k \tilde b_k$.
Therefore,
\begin{equation*}
\ex{\left. g_k \right|\theta_k} = \ex{u_k u_k\ts} \nabla\mu(\theta_k) + \ex{\left. b_k \right|\theta_k} = \frac{1}{p} \nabla\mu(\theta_k) + \ex{b_k |\theta_k},
\end{equation*}
since $u_k$ is uniformly distributed on the unit sphere $S^p$. Recall that $p>0$ is the dimensions of parameter vector, which does not harm the descent direction.

We now investigate the second moment, taking conditional expectation given $\theta_k$, we have
\begin{equation*}
\ex{\left. \norm{g_k} \right|\theta_k} = \ex{\left. \left(u_k\ts \nabla\mu(\theta_k)\right)^2 \right|\theta_k} + \ex{\left. \norm{\epsilon_k} \right|\theta_k} + \ex{\left. \tilde b_k^2 \right|\theta_k} + 2\ex{\left. \tilde b_k u_k\ts \nabla\mu(\theta_k) \right|\theta_k}.
\end{equation*}
Because $\norm[]{u_k}=1$ and \ref{A.model}, the last two terms are $\bigO[4]{c_k}$ and $\bigO{c_k}$, respectively. 
By simple algebra, we have
\begin{equation*}
    \ex{\left. \left(u_k\ts \nabla\mu(\theta_k)\right)^2 \right|\theta_k} = \frac{1}{p} \norm{\nabla\mu(\theta_k)},
\end{equation*}
where $\text{tr}(\cdot)$ is the trace of matrix. Furthermore, \ref{A.model}\ref{A.model.mu_US} and compactness of $\Theta$ implies that there exists some constant $L>0$ such that for any $\theta \in \Theta$, $\norm{\nabla\mu(\theta)} \leq L^2 \norm{\theta - \theta^*}$, since $\nabla\mu(\theta^*)=0$. By Taylor expansion around $\theta_k$, there is some $\bar\theta_k^\pm$ in the line segment between $\theta_k \pm c_k u_k$ and $\theta_k$, respectively, such that
\begin{equation*}
\begin{aligned}
    \ex{\left. \norm{\epsilon_k} \right|\theta_k}  &= \frac{\sigma^2(\theta_k)}{2\tau c_k^2} + \frac{1}{8\tau} \ex{\left. u_k\ts \left[\nabla^2\sigma^2(\bar\theta_k^+) + \nabla^2\sigma^2(\bar\theta_k^-)\right] u_k \right|\theta_k} \\
    &\leq \frac{\sigma^2(\theta_k)}{2\tau c_k^2} + \frac{1}{4\tau p} \sup_{\theta\in\Theta}\cBrac{\text{tr}\sBrac{\nabla^2\sigma^2(\theta)}},
\end{aligned}
\end{equation*}
where the inequality holds for $c_k$ sufficiently small due to the compactness of $\Theta$, and \ref{A.model}\ref{A.model.mu_US} and \ref{A.model}\ref{A.model.sigma}. Putting these together, we have
\begin{equation}  \label{eq.g2moment}
\ex{\norm{g_k}} \leq \frac{L^2}{p} \ex{\norm{\theta_k - \theta^*}} + \frac{1}{2\tau c_k^2}\ex{\sigma^2(\theta_k)} + \bigO[]{c_k^4 + c_k^2 + 1}.
\end{equation}

\section{Auxiliary Lemmas}
\setcounter{table}{0} 
\setcounter{figure}{0} 
\setcounter{equation}{0} 
\setcounter{theorem}{0}
\renewcommand{\thetable}{\thesection-\arabic{table}}
\renewcommand{\thefigure}{\thesection-\arabic{figure}}
\renewcommand{\theequation}{\thesection-\arabic{equation}}
\renewcommand{\thetheorem}{\thesection.\arabic{theorem}}

\begin{lemma}[Refinement to Lemma 3 in \citealt{hu2023}] \label{lemma.hu3}
Consider two real sequences $a_k=A(k+1)^{-\eta}$ and $b_k = \bigO[-\nu]{(k+1)}$, where $A\geq 1$, $\eta \in (1/2,1]$, $\nu \in [0,1)$ and $b_k>0$ for $k=0,1,\cdots$. Then for some integer $d\geq 1$,
\begin{equation*}
\sum_{k=0}^{n-1}\sBrac{\prod_{j=k+1}^{n-1}(1 - a_j)^d} a_k^d b_k = \bigO[]{a_n^{d-1}b_n}.
\end{equation*}
Furthermore, if $b_k = \boldsymbol\Omega\rBrac{(k+1)^{-\nu}}$ in addition, then $\nu<\left.dA^{d-1}\right/\rBrac{d+A^{d-1}}-\eta(d-1)$ implies that
\begin{equation*}
\sum_{k=0}^{n-1}\sBrac{\prod_{j=k+1}^{n-1}(1 - a_j)^d} a_k^d b_k = \boldsymbol\Omega\rBrac{a_n^{d-1}b_n}.
\end{equation*}
\end{lemma}
\begin{proof}{Proof}
By conditions, there exists an integer $n_0>0$ and a constant $C>0$ such that $\frac{1}{k+1}\leq a_k<1$ and $b_k \leq C (k+1)^{-\nu}$ whenever $k \geq n_0$, then we can partition the summation into finite part and tail part:
\begin{equation*}
\sum_{k=0}^{n-1}\sBrac{\prod_{j=k+1}^{n-1}(1 - a_j)^d} a_k^d b_k = \underbrace{\sum_{k=0}^{n_0-1}\sBrac{\prod_{j=k+1}^{n-1}(1 - a_j)^d} a_k^d b_k}_{\text{(finite part)}} + \underbrace{\sum_{k=n_0}^{n-1}\sBrac{\prod_{j=k+1}^{n-1}(1 - a_j)^d} a_k^d b_k}_{\text{(tail part)}}.
\end{equation*}
It follows from the inequality $1-x \leq e^{-x}$ that
\begin{equation*}
\begin{aligned}
\text{(finite part)} &= \prod_{j=1}^{n-1}(1 - a_j)^d \sum_{k=0}^{n_0-1} \sBrac{\prod_{j=1}^k (1 - a_j)}^{-d} a_k^d b_k \\
& \leq \exp\cBrac{-d\sum_{j=1}^{n-1} a_j} \absV{\sum_{k=0}^{n_0-1}\sBrac{\prod_{j=1}^k (1 - a_j)}^{-d} a_k^d b_k}
\end{aligned}
\end{equation*}
where the right-hand side goes to zero with order $\bigO[-d]{n}$ since $\eta \in (1/2, 1]$ and $n_0<\infty$.
\begin{equation*}
\exp\cBrac{-d\sum_{j=1}^{n-1} a_j} \leq \exp\cBrac{-d A \int_1^n x^{-\eta}dx} = 
\begin{cases}
\bigO[]{\exp\cBrac{-d A n^{1-\eta}}}, & \eta \in (1/2,1); \\
\bigO[-d A]{n} & \eta = 1.
\end{cases}
\end{equation*}

Now we use a sequence $\cBrac{x_n}_{n\geq n_0}$ to study the behavior of tail part.  First, by conditions, we have
\begin{equation}  \label{eq.tail_bound}
\text{(tail part)} \leq C \sum_{k=n_0}^{n-1}\sBrac{\prod_{j=k+1}^{n-1}(1 - a_j)^d} a_k^d (k+1)^{-\nu}.
\end{equation}
Let $x_{n_0}=0$, and for $n>n_0$, define
\begin{equation*}
x_n \eqdef \sum_{k=n_0}^{n-1}\sBrac{\prod_{j=k+1}^{n-1}(1 - a_j)^d} a_k^d (k+1)^{-\nu},
\end{equation*}
and a sequence of mappings $T_n(x) \eqdef (1-a_n)^d x + a_n^d (n+1)^{-\nu}$.  Then, it yields that
$x_{n+1} = T_n(x_n)$.
Note that $\absV{T_n(x) - T_n(y)} = (1-a_n)^d \absV{x-y}$ and $a_n < 1$ whenever $n\geq n_0$, thus $\cBrac{T_n}_{n\geq n_0}$ is a sequence of contraction mappings. Then, for each $T_n$, there is a unique fixed point
\begin{equation*}
x_n^* = \frac{a_n^d (n+1)^{-\nu}}{1-\rBrac{1-a_n}^d} = \frac{1}{d}a_n^{d-1} (n+1)^{-\nu} + e_n,
\end{equation*}
where
\begin{equation*}
e_n = -\frac{a_n^{d-1} \varpi_d(a_n) (n+1)^{-\nu}}{d(d+\varpi_d(a_n))}
\end{equation*}
with $\varpi_d(x) = \sum_{i=1}^{d-1} \binom{d}{i+1} \rBrac{-x}^i$.
Taking the specific form of $a_n$ by condition, it is easy to see $\absV{x_{n-1}^* - x_n^*} = \bigO[-\eta(d-1)-\nu-1]{n}$ because $\absV{\left.dx^{-\eta(d-1)-\nu}\right/dx}=\sBrac{\eta(d-1)+\nu}x^{-\eta(d-1)-\nu-1}$ and $e_n$ is a higher order infinitesimal term with respect to $a_n^{d-1} (n+1)^{-\nu}$.
Following the same procedures of the proof of Lemma 3 in \cite{hu2023}, we have
\begin{equation*}
\sum_{k=0}^{n-1}\sBrac{\prod_{j=k+1}^{n-1}(1 - a_j)^d} a_k^d b_k = \bigO[]{a_n^{d-1} b_n}.
\end{equation*}

Similar to \eqref{eq.tail_bound}, if $b_k = {\boldsymbol\Omega} \rBrac{(k+1)^{-\nu}}$, then for $n>n_0$, there is some constant $C'$ such that
\begin{equation*}
\text{(tail part)} \geq C' \sum_{k=n_0}^{n-1}\sBrac{\prod_{j=k+1}^{n-1}(1 - a_j)^d} a_k^d (k+1)^{-\nu}.
\end{equation*}
Since $\absV{x_n - x_n^*} = \bigO[-\eta(d-1)-\nu]{n}$ and $x_n^*>0$, it follows that
\begin{equation*}
\liminf_{n\to\infty}{\frac{x_n}{x_n^*}} \geq 1 - \limsup_{n\to\infty}{\frac{\absV{x_n-x_n^*}}{x_n^*}}.
\end{equation*}
It is worth noting that $\absV{x_n-x_n^*} \leq \frac{\eta(d-1)+\nu}{d-\eta(d-1)-\nu} n^{-\eta(d-1)-\nu}$ and $x_n^*=\frac{A^{d-1}}{d} (n+1)^{-\eta(d-1)-\nu}$.  Then, the condition
\begin{equation*}
\nu<\frac{dA^{d-1}}{d+A^{d-1}}-\eta(d-1)
\end{equation*}
implies that
\begin{equation*}
\limsup_{n\to\infty}{\frac{\absV{x_n-x_n^*}}{x_n^*}} \leq \frac{\eta(d-1)+\nu}{d-\eta(d-1)-\nu} \frac{d}{A^{d-1}} < 1,
\end{equation*}
which further leads to
\begin{equation*}
\liminf_{n\to\infty}{\frac{x_n}{x_n^*}} > 0.
\end{equation*}
It completes the proof of the lemma.
\end{proof}

\begin{lemma} \label{lemma.ewma}
Let $\{x_k\}_{k \geq 1}$ be a real sequence that $x_k \to 0$. Then, for any $\alpha \in (0,1)$,
\begin{equation*}
\sum_{k=1}^n \alpha^{n-k} x_k \to 0,~~~ n \to \infty.
\end{equation*}
\end{lemma}
\begin{proof}{Proof}
Let $\varepsilon > 0$, and $n_1 > 0$ such that $|x_k| < \frac{1-\alpha}{2} \varepsilon$ whenever $k \geq n_1$. Hence,
\begin{equation}
\label{ineq.1}
\begin{aligned}
\absV{ \sum_{k=1}^n \alpha^{n-k} x_k } &\leq  \absV{ \sum_{k=1}^{n_1-1} \alpha^{n-k} x_k } + \sum_{k=n_1}^n \alpha^{n-k} |x_k| \\
&< \alpha^n \absV{ \sum_{k=1}^{n_1-1} \alpha^{-k} x_k } + \sum_{k=0}^{n-n_1} \alpha^k \varepsilon \\
&< \alpha^n \absV{ \sum_{k=1}^{n_1-1} \alpha^{-k} x_k } + \frac{1}{2} \varepsilon,
\end{aligned}
\end{equation}
where the first inequality follows from triangle inequality and the last inequality follows from
\begin{equation*}
\frac{1-\alpha}{2}\sum_{k=0}^{n-n_1} \alpha^k \varepsilon = \frac{1 - \alpha^{n - n_1}}{2} \varepsilon < \frac{1}{2} \varepsilon.
\end{equation*}
Now choose a further $n_2 > n_1$ so that
\begin{equation}
\label{ineq.2}
\sup_{n \geq n_2} \alpha^n \absV{ \sum_{k=1}^{n_1-1} \alpha^{-k} x_k } < \frac{1}{2} \varepsilon.
\end{equation}
Combine (\ref{ineq.1}) and (\ref{ineq.2}), for all $n \geq n_2$,
\begin{equation*}
\absV{ \sum_{k=1}^n \alpha^{n-k} x_k } < \varepsilon.
\end{equation*}
The arbitrariness of $\varepsilon$ completes the proof of the lemma.
\end{proof}

\begin{lemma} \label{lemma.asymp_neg}
Suppose a sequence of random variables $\{\xi_k:~ k=0,1,\cdots\}$ satisfies that $\ex{\xi_k^2} \to 0$, and  $m(t) = \roundown{t/\gamma}$. Provided $T>0$, for any $s\in[0,T]$, then
\begin{equation*}
\sqrt{\gamma} \sum_{k=m(t)}^{m(t+s)-1} \xi_k \stackrel{\prob}{\longrightarrow} 0,~~~ \gamma \searrow 0.
\end{equation*}
\end{lemma}
\begin{proof}{Proof}
Note that $t-\gamma < \gamma m(t) \leq t$, then $\gamma m(t) \to t$ as $\gamma \searrow 0$. For any $\varepsilon>0$,
\begin{equation*}
\event{\absV{\sqrt{\gamma} \sum_{k=m(t)}^{m(t+s)-1} \xi_k} \geq \varepsilon}  \subseteq \event{\sqrt{\gamma} \sum_{k=m(t)}^{m(t+s)-1} |\xi_k| \geq \varepsilon}  \subseteq \bigcup_{k=m(t)}^{m(t+s)-1}\event{\sqrt{\gamma} |\xi_k| \geq \varepsilon}.
\end{equation*}
It follows from Bonferroni inequality and Markov inequality that
\begin{equation*}
\begin{aligned}
\prob \left(\absV{\sqrt{\gamma} \sum_{k=m(t)}^{m(t+s)-1} \xi_k} \geq \varepsilon \right) \leq&~ \sum_{k=m(t)}^{m(t+s)-1} \prob \left(\sqrt{\gamma} |\xi_k| \geq \varepsilon \right) \\
\leq&~ \frac{\gamma [m(t+s) - m(t)]}{\varepsilon^2} \frac{1}{m(t+s)-m(t)} \sum_{k=m(t)}^{m(t+s)-1} \ex{|\xi_k|^2},
\end{aligned}
\end{equation*}
where the right-hand side goes to 0 since $\gamma [m(t+s) - m(t)] \to s \leq T$, and that $\ex{|\xi_k|^2} \to 0$ implies the average goes to 0. This completes the proof.
\end{proof}

\section{Proofs} \label{sec.A.proofs}
\setcounter{table}{0} 
\setcounter{figure}{0} 
\setcounter{equation}{0} 
\setcounter{theorem}{0}
\renewcommand{\thetable}{\thesection-\arabic{table}}
\renewcommand{\thefigure}{\thesection-\arabic{figure}}
\renewcommand{\theequation}{\thesection-\arabic{equation}}
\renewcommand{\thetheorem}{\thesection.\arabic{theorem}}

\subsection{Proof of Proposition~\ref{thm.theta}} \label{subsec.A.theta}
See Theorem 1, \citealt{spall1992spsa} for the w.p.1 convergence.

Here we provide the convergence rate of root mean squared error of the parameter vector.
\noindent
Take $a_k = A/(k+1)^\eta$ and $c_k = C/(k+1)^\nu$ with some constants $A\geq 1$, $C>0$, then \ref{A.step}\ref{A.step.spsa} implies that $\eta \in (2/3, 1]$ and $(1-\eta)/2 < \nu < \eta - 1/2$ for all $k$. Let $\psi_k = \theta_k - \theta^*$. By result of (i) and \eqref{eq.theta*}, there is some $N>0$ such that $\theta_k \in \Theta$ whenever $k \geq N$, then
\begin{equation}
\begin{aligned}
\norm{\psi_{k+1}} &= \norm{\psi_k} - 2a_k\psi_k\ts g_k + a_k^2 \norm{g_k} \\
&\leq \norm{\psi_k} - 2a_k \psi_k\ts u_k u_k\ts \nabla\mu(\theta_k) - 2a_k\psi_k\ts (b_k + \epsilon_k) + a_k^2 \norm{g_k},
\label{eq.psi.B}
\end{aligned}
\end{equation}
where the inequality is by \eqref{eq.decomp_g}.
We now investigate the expectation of \eqref{eq.psi.B} by terms.
Note that \ref{A.model}\ref{A.model.mu_US} implies that the smallest eigenvalue of $\nabla^2\mu(\theta)$ is strictly positive for any $\theta \in \Theta$, and let $\lambda>0$ denote the infimum.
By \ref{A.model}\ref{A.model.mu_CD}\ref{A.model.mu_US}, which implies $\nabla\mu(\theta^*)=0$, a Taylor expansion of $\mu(\theta_k)$ around $\theta^*$ shows
\begin{align*}
\ex{\psi_k\ts u_k u_k\ts \nabla\mu(\theta_k)} &= \ex{\psi_k\ts u_k u_k\ts \nabla^2\mu(\tilde\theta_k) \psi_k} \\
&= \ex{\text{tr}\left(\ex[k]{u_k u_k\ts} \nabla^2\mu(\tilde\theta_k) \psi_k \psi_k\ts\right)} \\
&\geq \frac{2a_k \lambda}{p} \ex{\norm{\psi_k}},
\end{align*}
where $\tilde\theta_k$ lies in the line segment between $\theta_k$ and $\theta^*$.
Following the arguments in Subsection~\ref{subsec.estimators}, we have
\begin{equation*}
\begin{aligned}
-\ex{\psi_k\ts b_k} &\leq \ex{\absV{\tilde b_k} \psi_k\ts u_k} = \bigO{c_k}, \text{and} \\
\ex{\psi_k\ts \epsilon_k} &= \ex{\psi_k\ts \ex[k]{\epsilon_k}} = 0.
\end{aligned}
\end{equation*}
Now substitute \eqref{eq.g2moment} and the bounds derived above into the expectation of \eqref{eq.psi.B},
\begin{align*}
\ex{\norm{\psi_{k+1}}} \leq&~ \ex{\norm{\psi_k}} - 2a_k \ex{\psi_k\ts u_k u_k\ts \nabla^2\mu(\tilde\theta_k) \psi_k} \\
&\qquad - 2a_k \ex{\psi_k\ts (b_k + \epsilon_k)} + a_k^2 \ex{\norm{g_k}} \\
\leq&~ \left(1 - 2a_k \frac{\lambda}{p} + a_k^2 \frac{L^2}{p}\right) \ex{\norm{\psi_k}} + a_k \bigO{c_k} \\
&\qquad + a_k^2 \bigO[]{c_k^{-2} + 1 + c_k^4 +c_k^2}.
\end{align*}
By \ref{A.step}\ref{A.step.1}, there exists a further integer $n_0>0$ and some constant $\varrho > 0$ such that $2\lambda - a_k L^2 \geq 2p \varrho$, and hence 
\begin{equation*}
1 - 2a_k \frac{\lambda}{p} + a_k^2 \frac{L^2}{p} \leq 1 - 2\varrho a_k < 1,
\end{equation*}
whenever $k\geq n_0$. Therefore, we obtain
\begin{equation*}
\ex{\norm{\psi_{k+1}}} \leq \left(1 - 2\varrho a_k\right) \ex{\norm{\psi_k}} + a_k \sqrt{\ex{\norm{\psi_k}}} \bigO{c_k} + a_k^2 \bigO[]{c_k^{-2} + 1 + c_k^2 + c_k^4},
\end{equation*}
thus, there is some constant $B>0$ such that
\begin{equation*}
\begin{aligned}
\ex{\norm{\psi_{k+1}}} &\leq \left[\sqrt{(1-2\varrho a_k) \ex{\norm{\psi_k}}} + B a_k c_k^2\right]^2 + B^2 a_k^2 \left(c_k\inv + c_k\right)^2 \\
&\leq \left[(1-\varrho a_k) \sqrt{\ex{\norm{\psi_k}}} + B a_k c_k^2\right]^2 + B^2 a_k^2 \left(c_k\inv + c_k\right)^2,
\end{aligned}
\end{equation*}
where the last step is due to the inequality $\sqrt{1-x}\leq 1-x/2$ for $x \in [0,1]$.
Now for $k\geq n_0$, define a sequence of mappings 
\begin{equation*}
\mathcal{T}_k(x) \eqdef \sqrt{\left[(1-\varrho a_k) x + a_k c_k^2\right]^2 + B^2 a_k^2 \left(c_k\inv + c_k\right)^2},~~~ x\geq 0.
\end{equation*}
By induction, it is easy to observe that if we set $x_{n_0} = \sqrt{\ex{\norm{\psi_{n_0}}}}$, then $\sqrt{\ex{\norm{\psi_k}}} \leq x_k$ whenever $k \geq n_0$ for all $x_k$ generated by $x_{k+1} = \mathcal{T}_k(x_k),~k\geq n_0$. 
Furthermore, $\mathcal{T}_k$ is a contraction mapping with $\absV{\mathcal{T}(x) - \mathcal{T}(y)} < (1-\varrho a_k) \absV{x-y}$, since $d\sqrt{x^2 + z^2}/dx = x/\sqrt{x^2 + z^2} < 1$ for all $x\geq 0$ and any $z\neq 0$. 
Let $x_k^*$ denote the unique fixed point, i.e., $x_k^* = \mathcal{T}_k(x_k^*)$, which yields that
\begin{equation*}
\begin{aligned}
x_k^* &= \frac{B}{\varrho(2-\varrho a_k)} \left[(1-\varrho a_k) c_k^2 + \sqrt{c_k^4 + \varrho a_k (2 - \varrho a_k) \left(c_k\inv + c_k\right)^2}\right] \\
&< \frac{B}{\varrho} \left[2 c_k^2 + \sqrt{2\varrho a_k} \left(c_k\inv + c_k\right)\right] = \bigO[]{c_k^2 + \sqrt{a_k} \left(c_k\inv + c_k\right)},
\end{aligned}
\end{equation*}
where the inequality follows from $\sqrt{x + y}\leq \sqrt{x} + \sqrt{y}$ for $x,y\geq 0$ and that $0<1-\varrho a_k<1$ for $k\geq n_0$.
For the specific forms of $a_k$ and $c_k$ given in condition, it is not hard to verify that $\absV{x_k^* - x_{k+1}^*} = \bigO[-1-2\nu]{k} + \bigO[-1-\frac{\eta}{2}+\nu]{k}$.
Consequently, we have by the contraction property of $\mathcal{T}_k$ and lemma~\ref{lemma.hu3},
\begin{align*}
\absV{x_n - x_n^*} \leq&~ \absV{x_n - x_{n-1}^*} + \absV{x_{n-1}^* - x_n^*} = \absV{\mathcal{T}_{n-1}(x_{n-1}) - \mathcal{T}_{n-1}(x_{n-1}^*)} + \absV{x_{n-1}^* - x_n^*} \\
\leq&~ (1-\varrho a_{n-1}) \absV{x_{n-1} - x_{n-1}^*} + \absV{x_{n-1}^* - x_n^*} \leq \cdots \\
\leq&~ \prod_{k=n_0}^{n-1} (1-\varrho a_k) \absV{x_{n_0} - x_{n_0}^*} + \sum_{k=n_0}^{n-1} \left[\prod_{j=k+1}^{n-1} (1-\varrho a_j)\right] \absV{x_k^* - x_{k+1}^*} \\
\leq&~ \exp\left\{-\varrho \sum_{k=n_0}^{n-1} a_k\right\} \absV{x_{n_0} - x_{n_0}^*} + \sum_{k=n_0}^{n-1} \left[\prod_{j=k+1}^{n-1} (1-\varrho a_j)\right] a_k \frac{\absV{x_k^* - x_{k+1}^*}}{a_k} \\
=&~ \bigO[]{k^{-1+\eta -2\nu}} + \bigO[]{k^{-1+\frac{\eta}{2}+\nu}} = o(x_n^*).
\end{align*}
Finally,
\begin{equation*}
\sqrt{\ex{\norm{\psi_n}}} \leq \absV{x_n} \leq \absV{x_n - x_n^*} + \absV{x_n^*} = \bigO[]{c_n^2 + \sqrt{a_n} \left(c_n\inv + c_n\right)} = \bigO[-2\nu]{n} + \bigO[-\frac{\eta}{2}+\nu]{n},
\end{equation*}
which implies the desired result.

\subsection{Proof of lemma~\ref{lemma.mu_L2}}  \label{subsec.A.mu_L2}
{
Note that \eqref{eq.y_bar} can be decomposed as \eqref{eq.dec_y}, i.e.,
\begin{equation*}
\bar y_k = \mu(\theta_k) + \beta_k + \tilde\epsilon_k,
\end{equation*}
where
\begin{equation*}
\begin{aligned}
\beta_k &= \frac{c_k^2}{4}u_k\ts \sBrac{\nabla^2\mu(\tilde\theta_k^+) + \nabla^2\mu(\tilde\theta_k^-)} u_k, \\
\tilde\epsilon_k &= \frac{1}{2}\sBrac{\sigma(\theta_k + c_k u_k) \bar\epsilon_k^+ + \sigma(\theta_k - c_k u_k) \bar\epsilon_k^-}.
\end{aligned}
\end{equation*}
}
By \ref{A.model}\ref{A.model.mu_CD}\ref{A.model.mu_US} and Proposition~\ref{thm.theta}, $\beta_k = \bigO{c_k}$.  
Then, by noting that $\mu(\theta_k)-\mu^* = \psi_k\ts \nabla^2\mu(\tilde\theta_k) \psi_k$ for some $\tilde\theta_k$ in the line segment between $\theta_k$ and $\theta^*$, where $\psi_k = \theta_k - \theta^*$ as defined in the proof of Proposition~\ref{thm.theta}, we have
\begin{equation}
\ex[k]{\bar y_k - \mu^*} = \mu(\theta_k) - \mu^* + \ex[k]{\beta_k} = \bigO[]{\norm{\psi_k} + c_k^2}, \label{eq.yBar1moment}
\end{equation}
by \ref{A.model}\ref{A.model.mu_CD} and compact $\Theta$.
Now investigate the conditional second moment:
\begin{equation*}
\begin{aligned}
\ex[k]{\absV{\bar y_k - \mu^*}^2} =&~ \absV{\mu(\theta_k)-\mu^*}^2 + 2(\mu(\theta_k)-\mu^*) \ex[k]{\beta_k} + \ex[k]{\beta_k^2 + \tilde\epsilon_k^2} \\
\leq &~ \left(\sqrt{\absV{\mu(\theta_k)-\mu^*}^2} + \sqrt{\ex[k]{\beta_k^2}]}\right)^2 + \ex[k]{\tilde\epsilon_k^2},
\end{aligned}
\end{equation*}
where the inequality follows from H\"older's inequality. Note that $\tilde\epsilon_k$ shares the same second moment with $c_k u_k\ts \epsilon_k$, then
\begin{align*}
\ex[k]{\absV{\tilde\epsilon_k}^2} &= \frac{\sigma^2(\theta_k)}{2\tau} + \frac{c_k^2}{8\tau p} \text{tr} \left(\ex[k]{\nabla^2\sigma^2(\bar\theta_k^+) + \nabla^2\sigma^2(\bar\theta_k^-)}\right) \\
&= \frac{\sigma^2(\theta^*)}{2\tau} + \frac{1}{2\tau} \nabla\sigma^2(\bar\theta_k)\ts \psi_k + \bigO{c_k} \\
&\leq \frac{\sigma^2(\theta^*)}{2\tau} + \frac{1}{2\tau} \norm[]{\nabla\sigma^2(\bar\theta_k)} \norm[]{\psi_k} + \bigO{c_k} \\
&= \frac{\sigma^2(\theta^*)}{2\tau} + \bigO[]{\norm[]{\psi_k} + c_k^2},
\end{align*}
where the second identity follows from the Taylor expansion around $\theta^*$ with $\bar\theta_k$ in the line segment between $\theta_k$ and $\theta^*$ due to \ref{A.model}\ref{A.model.sigma}, and the inequality is due to H\"older's inequality. Then it yields
\begin{equation}
\ex[k]{\absV{\bar y_k - \mu^*}^2} = \frac{\sigma^2(\theta^*)}{2\tau} + \bigO[]{\mathcal{E}_k^2 + \norm[]{\psi_k} + c_k^2}, \label{eq.yBar2moment}
\end{equation}
where we have defined $\mathcal{E}_k \eqdef \norm{\psi_k} + c_k^2 \longrightarrow 0~ w.p.1$ by Proposition~\ref{thm.theta} and \ref{A.step}\ref{A.step.spsa}, which implies the first result. Now define $\varphi_k \eqdef \mu_k - \mu^*$ and rewrite \eqref{eq.mu} as
\begin{equation*}
\varphi_{k+1} = (1- \gamma) \varphi_k + \gamma (\bar y_k - \mu^*).
\end{equation*}
Taking conditional second moment in both sides of \eqref{eq.mu}, it yields
\begin{equation*}
\ex[k]{\absV{\varphi_{k+1}}^2} = (1-\gamma)^2 \absV{\varphi_k}^2 + 2\gamma (1-\gamma) \varphi_k \ex[k]{\bar y_k - \mu^*} + \gamma^2 \ex[k]{\absV{\bar y_k - \mu^*}^2}.
\end{equation*}
Substitute Eqs.(\ref{eq.yBar1moment}, \ref{eq.yBar2moment}) into the above equation, there exists some constant $K_\varphi>0$ such that
\begin{equation*}
\begin{aligned}
\ex[k]{\absV{\varphi_{k+1}}^2} \leq&~ (1-\gamma)^2 \varphi_k^2 + 2\gamma (1-\gamma) K_\varphi \varphi_k \mathcal{E}_k + \gamma^2 \left[\frac{\sigma^2(\theta^*)}{2\tau} + K_\varphi^2 \left(\mathcal{E}_k^2 + \norm[]{\psi_k} + c_k^2\right)\right] \\
=&~ \left[(1-\gamma) \varphi_k + \gamma K_\varphi \mathcal{E}_k\right]^2 + \gamma^2\left[\frac{\sigma^2(\theta^*)}{2\tau} + K_\varphi^2 \left(\norm[]{\psi_k} + c_k^2\right)\right].
\end{aligned}
\end{equation*}
Taking expectation, it follows from H\"older's inequality that
\begin{equation*}
\ex{\left((1-\gamma) \varphi_k + \gamma K_\varphi \mathcal{E}_k\right)^2} \leq \left[(1-\gamma) \sqrt{\ex{\absV{\varphi_k}^2}} + \gamma K_\varphi \sqrt{\ex{\absV{\mathcal{E}_k}^2}} \right]^2,
\end{equation*}
and hence, we have
\begin{equation*}
\ex{\absV{\varphi_{k+1}}^2} \leq \left[(1-\gamma) \sqrt{\ex{\absV{\varphi_k}^2}} + \gamma K_\varphi \sqrt{\ex{\absV{\mathcal{E}_k}^2}} \right]^2 + \gamma^2 H_k,
\end{equation*}
where $H_k = \frac{\sigma^2(\theta^*)}{2 \tau} + K_\varphi^2 \left(\ex{\norm[]{\psi_k}} + c_k^2\right)$.
Now we follow similar procedures in the proof of Proposition~\ref{thm.theta}(ii). Define a sequence of mappings
\begin{equation*}
\mathcal{M}_k(z) = \left(\left[(1-\gamma) z + \gamma K_\varphi \sqrt{\ex{\absV{\mathcal{E}_k}^2}} \right]^2 + \gamma^2 H_k\right)^{\frac{1}{2}},
\end{equation*}
which are contraction mappings with $\absV{\mathcal{M}_k(x) - \mathcal{M}_k(y)} < (1- \gamma)\absV{x-y}$ for $x,y\geq 0$.
It is easy to see by induction that, if we set $z_0 = \absV{\varphi_0}$, and $z_k$ are generated by $z_{k+1} = \mathcal{M}_k(z_k)$ for $k\geq 0$, $\sqrt{\ex{\absV{\varphi_k}^2}} \leq z_k$.
Denote the unique fixed point by $z_k^*$, solving $z_k^* = \mathcal{M}_k(z_k^*)$ yields that
\begin{equation*}
\begin{aligned}
z_k^* &= \frac{(1-\gamma) K_\varphi \sqrt{\ex{\absV{\mathcal{E}_k}^2}} + \sqrt{(1-\gamma)^2K_\varphi \ex{\absV{\mathcal{E}_k}^2} + \gamma (2- \gamma) H_k}}{2-\gamma} \\
&\leq \bigO[]{\sqrt{\ex{\absV{\mathcal{E}_k}^2}}} + \sqrt{\frac{\gamma}{2-\gamma} H_k} = \bigO[\frac{1}{2}]{\gamma},
\end{aligned}
\end{equation*}
where the inequality is due to that $\sqrt{x + y}\leq \sqrt{x} + \sqrt{y}$ for $x,y\geq 0$. By noting that $\ex{\absV{\mathcal{E}_k}^2} \longrightarrow 0$ and $H_k \longrightarrow \frac{\sigma^2(\theta^*)}{2 \tau}$, we obtain $z_k^* \longrightarrow \sqrt{\frac{\gamma \sigma^2(\theta^*)}{2 (2-\gamma) \tau}}$ and hence $\absV{z_k^* - z_{k+1}^*} \longrightarrow 0$. By the property of contraction mapping and Lemma~\ref{lemma.ewma}, it yields that $\absV{z_k - z_k^*} \longrightarrow 0$. Finally, we have
\begin{equation*}
\sqrt{\ex{\absV{\varphi_k}^2}} \leq z_k \leq \absV{z_k - z_k^*} + \absV{z_k^*} \longrightarrow \sqrt{\frac{\gamma \sigma^2(\theta^*)}{2 (2-\gamma) \tau}},
\end{equation*}
which completes the proof.

\subsection{Proof of Theorem \ref{thm.fclt}}  \label{subsec.A.fclt}

Plugging \eqref{eq.dec_y} in \eqref{eq.mu}, and expanding \eqref{def.sde_u} by iteration, for some $T>0$ and $s \in [0,T]$, we have
\begin{equation}
\label{eq.const_interpolation}
U^\gamma(t + s) - U^\gamma(t) = - \int_0^s U^\gamma(t + u)du + \rho^\gamma(t,s) + B^\gamma(t,s) + W^\gamma(t+s) - W^\gamma(t),
\end{equation}
where 
\begin{equation*}
\begin{aligned}
\rho^\gamma(t,s) &= [t+s-\gamma m(t+s)]U^\gamma_{m(t+s)}, \\
B^\gamma(t,s) &= \sqrt{\gamma} \sum_{k=m(t)}^{m(t+s)-1} [\mu(\theta_k) - \mu^* + \beta_k], \\
W^\gamma(t) &= \sqrt{\gamma}\sum_{k=0}^{m(t)-1} \tilde\epsilon_k.
\end{aligned}
\end{equation*}

Below we show in Lemmas~\ref{lemma.mu_rho}-\ref{lemma.mu_bias} that the terms $\rho^\gamma(t,s)$ and $B^\gamma(t,s)$ are asymptotic negligible (as $\gamma \searrow0$) so that they would not affect the asymptotic behavior of $\{U^\gamma(t+\cdot), W^\gamma(t+\cdot)\}$ on $[0,T]$. The limit of $W^\gamma(\cdot)$ is proved to be a Wiener process in Lemma~\ref{lemma.mu_fclt}. 
\begin{lemma} \label{lemma.mu_rho}
Suppose Assumptions \ref{A.model}-\ref{A.step} hold. Then provided $T>0$, for any $t>0$ and every $s \in [0,T]$, as $\gamma \searrow 0$, $\rho^\gamma(t,s) \stackrel{\prob}{\longrightarrow} 0$.
\end{lemma}
\begin{proof}{Proof}
It is a corollary of Lemma~\ref{lemma.mu_L2}. Note that $t-\gamma < \gamma m(t) \leq t$, then $\gamma m(t) \to t$ as $\gamma \searrow 0$. Thus,
\begin{equation*}
\rho^\gamma(t,s) = [\mu_{m(t+s)} - \mu^*] \bigO[]{\sqrt{\gamma}}.
\end{equation*}
Therefore, 
\begin{equation*}
\ex{\absV{\rho^\gamma(t,s)}} = \left[\sup_k \ex{\absV{\mu_k}} + \absV{\mu^*}\right] \bigO[]{\sqrt{\gamma}} \to 0,
\end{equation*}
as $\gamma \searrow 0$. It is trivial that convergence in $L^1$ implies convergence in probability, which completes the proof of the lemma.
\end{proof}

\begin{lemma} \label{lemma.mu_bias}
Suppose Assumptions \ref{A.model}-\ref{A.step} hold. Then provided $T>0$, for any $t>0$ and every $s \in [0,T]$, as $\gamma \searrow 0$, $B^\gamma(t,s) \stackrel{\prob}{\longrightarrow} 0$.
\end{lemma}
\begin{proof}{Proof}
It follows from \ref{A.model}\ref{A.model.mu_CD} and Proposition~\ref{thm.theta} that $\ex{\absV{\mu(\theta_k) - \mu^*}^2} \to 0$. In addition, $\ex{\absV{\beta_k}^2} = \bigO[4]{c_k} \to 0$ by \ref{A.model}\ref{A.model.mu_CD} and \ref{A.step}\ref{A.step.spsa}, which leads to
\begin{equation*}
\ex{\absV{\mu(\theta_k) - \mu^* + \beta_k}^2} \leq \left[\sqrt{\ex{\absV{\mu(\theta_k) - \mu^*}^2}} + \sqrt{\ex{\absV{\beta_k}^2}}\right]^2 \to 0.
\end{equation*}
Therefore, it follows from Lemma~\ref{lemma.asymp_neg} that
\begin{equation*}
B^\gamma(t,s) = \sqrt{\gamma} \sum_{k=m(t)}^{m(t+s)-1} [\mu(\theta_k) - \mu^* + \beta_k] \stackrel{\prob}{\longrightarrow} 0,
\end{equation*}
for every $t \in [0,T]$ for some $T>0$. It completes the proof of the lemma.
\end{proof}

\begin{lemma} \label{lemma.mu_fclt}
Suppose Assumptions \ref{A.model}-\ref{A.step} hold. Then for any $t>0$, as $\gamma \searrow 0$, $W^\gamma(t) \stackrel{d}{\longrightarrow} \frac{\sigma(\theta^*)}{\sqrt{2\tau}} W(t)$, where $W(t)$ is the standard Wiener process.
\end{lemma}
\begin{proof}{Proof}
Define
\begin{equation*}
\begin{aligned}
\xi_k &\eqdef \sqrt{\tau/2}(\bar\epsilon_k^+ + \bar\epsilon_k^-),~ \text{and} \\
\varsigma(\theta) &\eqdef \sigma(\theta) - \sigma(\theta^*),
\end{aligned}
\end{equation*}
then $\ex{\xi_k}=0,~\ex{\xi_k^2}=1$, and
\begin{equation*}
\tilde\epsilon_k - \frac{\sigma(\theta^*)}{\sqrt{2\tau}}\xi_k = \frac{\varsigma(\theta_k + c_k u_k) \bar\epsilon_k^+ + \varsigma(\theta_k - c_k u_k) \bar\epsilon_k^-}{2},
\end{equation*}
and hence, by Taylor expansion around $\theta^*$,
\begin{equation*}
\begin{aligned}
\ex{\absV{\tilde\epsilon_k - \frac{\sigma(\theta^*)}{\sqrt{2\tau}} \xi_k}^2} &= \frac{\varsigma^2(\theta_k + c_k u_k) + \varsigma^2(\theta_k - c_k u_k)}{4\tau} \\ 
&= \frac{1}{2\tau} \varsigma^2(\theta_k) + \frac{c_k^2}{4\tau} \ex{u_k\ts \left[\varsigma^2(\bar\theta_k^+) + \varsigma^2(\bar\theta_k^-)\right] u_k} \longrightarrow 0,
\end{aligned}
\end{equation*}
since \ref{A.model}\ref{A.model.sigma}, \ref{A.u} and \ref{A.step}\ref{A.step.spsa}.
Then
\begin{equation}
\label{eq.W_decompos}
W^\gamma(t ) = \sqrt{\gamma}\sum_{k=0}^{m(t)-1} \left(\tilde\epsilon_k - \frac{\sigma(\theta^*)}{\sqrt{2\tau}} \xi_k\right) + \sigma(\theta^*) \sqrt{\frac{\gamma}{2\tau}}\sum_{k=0}^{m(t)-1} \xi_k,
\end{equation}
where the first term on right-hand side converges to 0 in probability by lemma~\ref{lemma.asymp_neg}.
It is trivial by the Slutsky's theorem and CLT that
\begin{equation*}
\sqrt{\gamma} \sum_{k=0}^{m(t)-1} \xi_k = \sqrt{\gamma m(t)} \frac{1}{\sqrt{m(t)}} \sum_{k=0}^{m(t)-1} \xi_k \stackrel{d}{\longrightarrow} W(s),
\end{equation*}
since $\gamma m(t) \to t$ as $\gamma \searrow 0$. It completes the proof of the lemma.
\end{proof}

\begin{proof}{\bf Proof of Theorem \ref{thm.fclt}}  
Firstly, we show that $\{U^\gamma_k\}$ is tight (in $\gamma$) for sufficiently large $k$. For some $\delta>0$, it follows from Markov inequality that
\begin{equation*}
\prob \left(\absV{U^\gamma_k} \geq \delta\right) = \prob \left( \frac{\absV{\mu_k-\mu^*}}{\sqrt{\gamma}} \geq \delta \right) \leq \frac{\ex{\absV{\mu_k - \mu^*}^2}}{\gamma \delta^2},
\end{equation*}
Then, by Lemma~\ref{lemma.mu_L2}, we have
\begin{equation*}
\limsup_k \prob \left(\absV{U^\gamma_k} \geq \delta\right) = \bigO[]{\delta^{-2}},
\end{equation*}
which implies the tightness of $\{U^\gamma_k\}$ (in $\gamma$) for sufficiently large $k$.  Therefore, the sequence $\{U^\gamma_k\}$ is weakly convergent as $\gamma \searrow 0$.  

Now let us return to \eqref{eq.const_interpolation}. Fixed $t > 0$, $U^\gamma(t) = U^\gamma_{m(t)}$ by definition, so that $\{U^\gamma(t)\}$ is also tight for sufficiently small $\gamma$.  Then, it follows from Lemmas~\ref{lemma.mu_rho}-\ref{lemma.mu_fclt} that $\{U^\gamma(t+\cdot),W^\gamma(t+\cdot)\}$ converges weakly to some limit $\{U(\cdot),W(\cdot)\}$ that satisfies
\begin{equation*}
U(t+s) = U(t) - \int_0^s U(t+u)du + \frac{\sigma(\theta^*)}{\sqrt{2\tau}}[W(t+s) - W(t)],
\end{equation*}
which is equivalent to SDE \eqref{sde}.  It completes the proof.
\end{proof}

\subsection{Proof of Theorem \ref{thm.consist}}  \label{subsec.A.consist}
Rewriting \eqref{eq.sCons} as
\begin{equation*}
\begin{aligned}
\frac{(e^\gamma U^\gamma_{k+1}-U^\gamma_k)^2}{\gamma} =&~ \frac{\left[e^\gamma(\mu_{k+1}-\mu_k) + (e^\gamma - 1)(\mu_k - \mu^*)\right]^2}{\gamma^2} \\
=&~ \left[e^\gamma(\bar y_k - \mu_k) + \frac{e^\gamma - 1}{\gamma}(\mu_k - \mu^*)\right]^2 \\
=&~ e^{2\gamma}(\bar y_k - \mu_k)^2 + \frac{2e^\gamma(e^\gamma - 1)}{\gamma}(\bar y_k - \mu_k)(\mu_k - \mu^*) + \frac{(e^\gamma - 1)^2}{\gamma^2}(\mu_k - \mu^*)^2,
\end{aligned}
\end{equation*} 
we then focus on the term
\begin{equation}
\label{eq.estOU_var}
(\bar y_k -\mu_k)^2 = \frac{(U^\gamma_{k+1} - e^{-\gamma} U^\gamma_k)^2}{\gamma} + \frac{2(e^{-\gamma} - 1)}{\gamma}(\bar y_k - \mu_k)(\mu_k - \mu^*) - \frac{(e^{-\gamma} - 1)^2}{\gamma^2}(\mu_k - \mu^*)^2.
\end{equation}
It is obvious that if both
\begin{equation*}
\begin{aligned}
(\bar y_k - \mu_k)(\mu_k - \mu^*) &\stackrel{\prob}{\longrightarrow} 0, ~\text{and} \\
(\mu_k - \mu^*)^2 &\stackrel{\prob}{\longrightarrow} 0
\end{aligned}
\end{equation*}
hold as $\gamma \searrow 0,~k\to\infty$, then our variance estimator \eqref{eq.var}, which is equivalent to
\begin{equation*}
v_n = \frac{1}{n} \sum_{k=0}^{n-1} (\bar y_k -\mu_k)^2 + \frac{v_0}{n},
\end{equation*}
is consistent.  
Note that the (biased) sample variance \eqref{eq.sampleVar}, defined in subsection~\ref{subsec.estimators},
\begin{align*}
\hat \sigma^2_n &= \frac{1}{n}\sum_{k=0}^{n-1} (y_k - \bar\mu_n)^2 \\
&= \frac{n-1}{n} \hat \sigma^2_{n-1} + \frac{n-1}{n}(\bar\mu_n - \bar\mu_{n-1})^2 + \frac{1}{n} (y_{n-1} - \bar\mu_n)^2 \\
&= \hat \sigma^2_{n-1} + \frac{1}{n} \left[\frac{n-1}{n} (y_{n-1} - \bar\mu_{n-1})^2 - \hat \sigma^2_{n-1}\right] \\
&= \frac{1}{n}\sum_{k=0}^{n-1} \frac{k}{k+1} (y_k - \bar\mu_k)^2,
\end{align*}
hence $\hat \sigma^2_n$ can be interpreted as a variance estimator that weighted less on the early observations $(y_k - \bar\mu_k)^2$.
Therefore, for some positive non-decreasing real sequence $\{r_k\}_{k\geq 0}$ that $r_k \nearrow 1$, e.g., $r_k = k/(k+1)$, it follows from Lemma~\ref{lemma.mu_L2} that
\begin{equation*}
r_k(\bar y_k - \mu_k)^2 \stackrel{\prob}{\longrightarrow} \frac{\sigma^2(\theta^*)}{2\tau},~~~ k\to\infty,
\end{equation*}
which supports the improved estimator of the variance given by \eqref{eq.var*}.

\begin{proof}{\bf Proof of Theorem \ref{thm.consist}}
Rewrite \eqref{eq.estOU_var} as
\begin{equation*}
\begin{aligned}
(\bar y_k - \mu_k)^2 =&~ \frac{(U^\gamma_{k+1} - e^{-\gamma} U^\gamma_k)^2}{\gamma} - \left[\frac{2(e^{-\gamma} - 1)}{\gamma} + \frac{(e^{-\gamma} - 1)^2}{\gamma^2}\right] (\mu_k - \mu^*)^2 \\
&+ \frac{2(e^{-\gamma} - 1)}{\gamma} (\bar y_k - \mu^*)(\mu_k - \mu^*).
\end{aligned}
\end{equation*}
It follows from Markov inequality that for every $\varepsilon>0$,
\begin{equation*}
\prob \left(\left|\frac{1}{n}\sum_{k=0}^{n-1} r_k (\mu_k - \mu^*)^2\right| \geq \varepsilon \right) \leq \frac{1}{n\varepsilon}\sum_{k=0}^{n-1} r_k \ex{(\mu_k - \mu^*)^2},
\end{equation*}
and with H\"older's inequality that
\begin{equation*}
\begin{aligned}
\prob \left(\left|\frac{1}{n}\sum_{k=0}^{n-1} r_k (\bar y_k - \mu^*)(\mu_k - \mu^*)\right| \geq \varepsilon \right) \leq&~ \frac{1}{n\varepsilon}\sum_{k=0}^{n-1} r_k \ex{\absV{(\bar y_k - \mu^*)(\mu_k - \mu^*)}} \\
\leq&~ \frac{1}{n\varepsilon}\sum_{k=0}^{n-1} r_k \sqrt{\ex{(\bar y_k - \mu^*)^2} \ex{(\mu_k - \mu^*)^2}}.
\end{aligned}
\end{equation*}
{
By \eqref{eq.sCons}, it suffices to show that the right-hand sides of both inequalities go to 0 as $\gamma \searrow 0$.
By Theorem~\ref{thm.fclt}, we have
\begin{equation*}
\lim_{\gamma\searrow 0} U_{m(t)}^\gamma \stackrel{d}{=} U(t), ~~~ \forall t > 0.
\end{equation*}
Note that by Lemma~\ref{lemma.mu_L2}, we have for all $t > 0$,
\begin{equation*}
\limsup_{\gamma \searrow 0} \gamma \ex{\rBrac{U_{m(t)}^\gamma}^2} < \infty,
\end{equation*}
then
\begin{equation*}
\gamma \ex{\rBrac{U_{m(t)}^\gamma}^2} = \gamma \ex{U(t)^2} + o(\gamma),
\end{equation*}
where a sequence (indexed by $\gamma$) $a_\gamma = o(\gamma)$ means that $\lim_{\gamma \searrow 0} a_\gamma/\gamma = 0$.
Therefore,
\begin{align*}
\frac{1}{n}\sum_{k=0}^{n-1}\ex{\rBrac{\mu_k - \mu^*}^2} &= \frac{\gamma}{n}\sum_{k=0}^{n-1}\ex{U(k\gamma)^2} + o(\gamma) \\
&= \frac{1}{n}\sum_{k=0}^{n-1} \sBrac{e^{-2k\gamma}\rBrac{\mu_0 - \mu^*}^2 + \frac{\gamma\sigma^2(\theta^*)}{4\tau}\rBrac{1 - e^{-2k\gamma}}} + o(\gamma) \\
&= \frac{\gamma\sigma^2(\theta^*)}{4\tau} + \frac{1}{n}\sum_{k=0}^{n-1} e^{-2k\gamma} \sBrac{\rBrac{\mu_0 - \mu^*}^2 - \frac{\gamma\sigma^2(\theta^*)}{4\tau}} + o(\gamma).
\end{align*}
By Riemann-Stieltjes integration, we have
\begin{equation*}
\frac{1}{n}\sum_{k=0}^{n-1} e^{-2k\gamma} \leq \frac{1}{n}\rBrac{1+ \int_{0}^{n-1} e^{-2\gamma x} dx} = \frac{1}{2n\gamma}\rBrac{1+2\gamma-e^{-2(n-1)\gamma}}.
\end{equation*}
The right-hand side goes to 0 as $\gamma\searrow 0$ with $n > -\frac{1}{2\gamma}\log\gamma$.}
 
Combined with \eqref{eq.sCons}, we obtain
\begin{equation*}
\frac{1}{n}\sum_{k=0}^{n-1} r_k(\bar y_k - \mu_k)^2 \stackrel{\prob}{\longrightarrow} \frac{\sigma^2(\theta^*)}{2\tau},~~~ {n > -\frac{1}{2\gamma}\log\gamma}, \gamma\searrow 0.
\end{equation*}
Finally, $v_n$ generated by \eqref{eq.var*} is equivalent to
\begin{equation*}
v_n = \frac{1}{n}\sum_{k=0}^{n-1} r_k(\bar y_k - \mu_k)^2,
\end{equation*}
which is a consistent estimator of $\sigma^2(\theta^*)/2$.
\end{proof}

\subsection{Proof of Proposition \ref{prop.clt}}  \label{subsec.A.clt}
Following the arguments in Subsection \ref{subsec.e1}, it suffices to show that both \eqref{eq.Lyapunov_condition} and \eqref{eq.var_cvg} holds.  It follows from Assumption~\ref{A.model}\ref{A.model.sigma} and condition ({\bf C}) that for some constant $d\leq 1$ and $K>0$,
\begin{equation*}
\begin{aligned}
\sum_{k=1}^{n} \ex[k]{\absV{\xi_{nk}}^{2+d}} &= {V_n}^{-1-d/2} \sum_{k=1}^{n} \gamma_k^{2+d} \prod_{j=k+1}^n (1-\gamma_j)^{2+d} \ex[k]{\absV{\tilde\epsilon_k}^{2+d}} \\
&\leq K {V_n}^{-1-d/2} \sum_{k=1}^{n} \gamma_k^{2+d} \prod_{j=k+1}^n (1-\gamma_j)^{2+d} \rBrac{\sigma^{2+d}(\theta_k) + c_k}.
\end{aligned}
\end{equation*}
By Lemma~\ref{lemma.hu3}, we have
\begin{equation*}
\sum_{k=1}^{n} \gamma_k^{2+d} \prod_{j=k+1}^n (1-\gamma_j)^{2+d} \rBrac{\sigma^{2+d}(\theta_k) + c_k} = \rBrac{1+\norm[]{\theta_n-\theta^*}+c_n} \bigO[d+1]{\gamma_n},
\end{equation*}
then it is a sufficient condition for \eqref{eq.Lyapunov_condition} that 
$V_n = \boldsymbol\Theta(\gamma_n)$, 
which yields
\begin{equation*}
\sum_{k=1}^{n} \ex[k]{\absV{\xi_{nk}}^{2+d}} = \bigO[d/2]{\gamma_n}.
\end{equation*}

Similarly, rewrite \eqref{eq.var_cvg} as
\begin{equation*}
\sum_{k=1}^{n} \ex[k]{\xi_{nk}^2} = {V_n}\inv \sum_{k=1}^{n} \gamma_k^2 \prod_{j=k+1}^n (1-\gamma_j)^2 \ex[k]{\absV{\tilde\epsilon_k}^2},
\end{equation*}
where $\ex[k]{\absV{\tilde\epsilon_k}^2} = \sigma^2(\theta^*)/(2\tau) + \bigO[]{\norm[]{\theta_k-\theta^*}} + \bigO[]{c_k}$.  Thus, we have to show that $\sum_{k=1}^{n} \gamma_k^2 \prod_{j=k+1}^n (1-\gamma_j)^2$ is on the same order of $\gamma_n$.  The upper bound is immediately supported by Lemma~\ref{lemma.hu3}, while the asymptotic lower bound on the same order needs in addition that $\delta<2\gamma/(2+\gamma)$.  Then, a stronger version of \eqref{eq.var_cvg} holds that
\begin{equation*}
\sum_{k=1}^{n} \ex[k]{\xi_{nk}^2} \longrightarrow \frac{\sigma^2(\theta^*)}{2\tau}~ w.p.1.
\end{equation*}
Then the CLT holds.

Note that
\begin{equation*}
\begin{aligned}
\frac{\ex[k]{\absV{\xi_{nk}}^{2+d}}}{\ex[k]{\xi_{nk}^2}} &= {V_n}^{-d/2} \gamma_k^d \prod_{j=k+1}^n (1-\gamma_j)^d \frac{\ex[k]{\absV{\tilde\epsilon_k}^{2+d}}}{\ex[k]{\tilde\epsilon_k^2}} \\
&\leq K {V_n}^{-d/2} \gamma_k^d \prod_{j=k+1}^n (1-\gamma_j)^d \\
&\leq K {V_n}^{-d/2} \max_{1\leq k\leq n} \gamma_k^d \prod_{j=k+1}^n (1-\gamma_j)^d \\
&= \bigO[d/2]{\gamma_n}
\end{aligned}
\end{equation*}
Then, by Theorem 2.1 in \citet{fan2019}, the Berry-Esseen bound for the martingale CLT follows, i.e.,
\begin{equation*}
\sup_{z\in\reals}\absV{\prob\cBrac{\sqrt{2\tau} S_n \leq z\sigma(\theta^*)}-\Phi\rBrac{z}} = \bigO[d/2]{\gamma_n}.
\end{equation*}

\section{Benchmarks} \label{sec.A.benchmarks}
\setcounter{table}{0} 
\setcounter{figure}{0} 
\setcounter{equation}{0} 
\setcounter{theorem}{0}
\renewcommand{\thetable}{\thesection-\arabic{table}}
\renewcommand{\thefigure}{\thesection-\arabic{figure}}
\renewcommand{\theequation}{\thesection-\arabic{equation}}
\renewcommand{\thetheorem}{\thesection.\arabic{theorem}}

\subsection{Benchmarks in Section \ref{sec.compr}}  \label{subsec.benchmarks}
This subsection briefly introduces some benchmarks: four-point method \citep{wu2022}, multi-time-scale method \citep{borkar2008} and forward FD method.  For notational simplicity and consistency, we use diacritics to distinguish the estimates generated by different methods.  
\subsubsection*{Four-point method.}
Samples $y^{(\pm)}_{ki}$ ($i=1, \cdots, \tau_1$) and $y^{(h\pm)}_{kj}$ ($j=1, \cdots, \tau_2$) are drawn as realizations of $Y(\theta_k \pm c_k u_k)$ and $Y(\theta_k \pm h c_k u_k)$, respectively, where $h>1$ is a user-specified hyperparameter such that $h=\tau_1/\tau_2$ and $\tau_1 + \tau_2 = \tau$.  Then the sample averages are given by
\begin{equation*}
y_k^{(\pm)} = \frac{1}{\tau_1} \sum_{i=1}^{\tau_1} y^{(\pm)}_{ki}, ~~~
y_k^{(h\pm)} = \frac{1}{\tau_2} \sum_{j=1}^{\tau_2} y^{(h\pm)}_{kj},
\end{equation*}
and the gradient and performance can be estimated as
\begin{equation*}
\begin{aligned}
\hat g_k &= \frac{-y^{(h+)}_k + h^3 y^{(+)}_k - h^3 y^{(-)}_k + y^{(h-)}_k}{2h^2(h-1)c_k}u_k,\\
\hat y_k &= \frac{-y^{(h+)}_k + h^2 y^{(+)}_k + h^2 y^{(-)}_k - y^{(h-)}_k}{2(h^2-1)}. \\
\end{aligned}
\end{equation*}
The parameter vector is still updated via \eqref{eq.theta} with $g_k$ being replaced by $\hat g_k$.  For inference, it uses the offline sample average of $\hat y_k~k=0,1,\cdots,n$, which is equivalent to an online estimator generated by 
\begin{equation}
\hat\mu_{k+1} = \hat\mu_k + \frac{1}{k+1}(\hat y_k - \hat \mu_k).  \label{eq.mu*}
\end{equation}
The variance is estimated by the sample variance of $\{y_i,~ i=I,I+1,\cdots,2\tau n\}$ for some $I\in \naturals_+$,  which does not adapt to the online environment. Here, $y_i$ ($i=2k\tau+1, k\tau +1,\cdots, 2(k+1)\tau$) represents a random permutation of $\{y^{(\pm)}_{ki}, y^{(h\pm)}_{kj},~ i=1,2,\cdots,\tau_1,~ j=1,2,\cdots, \tau_2\}$.
Following the arguments in Subsection~\ref{subsec.consistency}, it is reasonable to use \eqref{eq.var*} to estimate the variance for both methods.  Conventionally, denote $\hat v_n$ the variance estimator for the four-point method.  According to \cite{wu2022}, we know that
\begin{equation*}
\frac{h^2 - 1}{h^2 + 1} \cdot \frac{\hat \mu_n-\mu^*}{\sigma(\theta^*)/\sqrt{2\tau n}} \stackrel{d}{\longrightarrow} \Normal{1},~~~ n\to\infty,
\end{equation*}
and one can prove that
\begin{equation*}
\ex{\left(\hat y_k - \hat\mu_k\right)^2} \longrightarrow \left(\frac{h^2 + 1}{h^2 - 1}\right)^2 \frac{\sigma^2(\theta^*)}{2\tau},~~~ k\to\infty,
\end{equation*}
which further implies that
\begin{equation*}
\hat v_n \stackrel{\prob}{\longrightarrow} \left(\frac{h^2 + 1}{h^2 - 1}\right)^2 \frac{\sigma^2(\theta^*)}{2\tau},~~~ n\to\infty.
\end{equation*}
Then the associated normalized estimator is
\begin{equation*}
\frac{(h^2 - 1)}{(h^2 + 1)} \frac{\hat{\mu}_n - \mu^*}{\sqrt{\hat v_n / n}},
\end{equation*}

\subsubsection*{Multi-time-scale method.}
The optimization procedure of the multi-time-scale method employs SPSA as Algorithm~\ref{alg.spsasi} does, while the  recursive equation of performance estimation merely replaces $\gamma$ by a decreasing positive sequence $\cBrac{\gamma_k}$. Specifically, \eqref{eq.mu} and \eqref{eq.var} are replaced by
\begin{equation*}
\begin{aligned}
\tilde{\mu}_{k+1} &= \tilde\mu_k + \gamma_k(\bar{y}_k - \tilde\mu_k), \\
\tilde{v}_{k+1} &= \tilde{v}_k + \zeta_k \sBrac{\rBrac{\bar{y}_k - \tilde\mu_k}^2 - \tilde{v}_k},
\end{aligned}
\end{equation*}
respectively, where $\cBrac{\gamma_k}$ and $\cBrac{\zeta_k}$ satisfy the conditions $\min\rBrac{\sum_k \gamma_k, \sum_k \zeta_k} =\infty,~ \sum_{k} (\gamma_k^2 + \zeta_k^2)<\infty$ and $a_k = o(\gamma_k),~ \gamma_k = o(\zeta_k)$ \citep{borkar2008}.

Previous result \citep[Theorem 2.1, Section 10.2]{kushner2003} implies that, under proper conditions, for sufficiently large $n$,
\begin{equation*}
\frac{\tilde{\mu}_n - \mu^*}{\sqrt{\gamma_n}} \approx \Normal{\frac{\sigma^2(\theta^*)}{2\tau}}.
\end{equation*}
Thus, the associated normalized estimator 
\begin{equation*}
\frac{\tilde\mu_n - \mu^*}{\sqrt{\gamma_n \tilde{v}_n}} 
\end{equation*}
is valid to construct CIs if the variance estimator is consistent.

\subsubsection*{Forward SPSA method.}
Instead of realizations $y^-_{kj},j=1,2,\cdots,\tau$ of $Y(\theta_k - c_k u_k)$ in Algorithm~\ref{alg.spsasi}, forward SPSA utilizes samples of $Y(\theta_k)$ of which the batch average is denoted by $\check{y}_k$ in the $k$th iteration.  It is naive to avoid the bias induced by perturbation via using only the samples of $Y(\theta_k)$ to estimate the performance, even though it neglects the information provided by perturbed performances.
Specifically, replace \eqref{eq.gradient} by
\begin{equation*}
\check{g}_k = \frac{y^+_k - \check{y}_k}{c_k}u_k,
\end{equation*}
then substitute $g_k$ in \eqref{eq.theta} by $\check{g}_k$.  Instead of \eqref{eq.mu}, the associated performance estimates are generated by
\begin{equation*}
\check{\mu}_{k+1} = \check{\mu}_k + \frac{1}{k+1} \rBrac{\check{y}_k - \check{\mu}_k}.
\end{equation*}
It corresponds to the variance estimator 
\begin{equation*}
\check{v}_{k+1} = \check{v}_k + \frac{1}{k+1}\sBrac{\rBrac{\check{y}_k - \check\mu_k}^2 - \check{v}_k},
\end{equation*}
and the associated normalized estimator
\begin{equation*}
\frac{\check\mu_n - \mu^*}{\sqrt{\check{v}_n/n}}.
\end{equation*}

\subsection{Benchmarks in Section \ref{sec.exprm}}
Apart from the algorithms introduced in Subsection~\ref{subsec.benchmarks}, we briefly describe two supplementary methods without theoretical guarantees here.

\subsubsection*{Ordinary SPSA.}
It is a naive algorithm for the dual tasks. It means to use SPSA for optimization, and simply take the sample average $\frac{1}{n}\sum_{k=0}^{n-1}\bar{y}_k$ as the performance estimator.  It replaces \eqref{eq.mu} and \eqref{eq.var} by
\begin{equation*}
\begin{aligned}
\bar{\mu}_{k+1} &= \bar{\mu}_k + \frac{1}{k+1}\rBrac{\bar{y}_k - \bar{\mu}_k}, \\
\bar{v}_{k+1} &= \bar{v}_k + \frac{1}{k+1} \sBrac{\rBrac{\bar{y}_k - \bar{\mu}_k}^2 - \bar{v}_k},
\end{aligned}
\end{equation*}
and the associated normalized estimator is
\begin{equation*}
\frac{\bar{\mu}_n - \mu^*}{\sqrt{\bar{v}_n/n}}.
\end{equation*}
It has been shown in Subsection~\ref{subsec.estimators} that this method fails in statistical inference for the optimal performance.   

\subsubsection*{Multi-time-scale four-point method.}
Here, we brutely combine the multi-time-scale SA with the four-point method.  Intuitively, we can use the four-point method for optimization and then replace \eqref{eq.mu} and \eqref{eq.var} by
\begin{equation*}
\begin{aligned}
\acute{\mu}_{k+1} &= \acute{\mu}_k + \gamma_k\rBrac{\hat{y}_k - \acute{\mu}_k}, \\
\acute{v}_{k+1} &= \acute{v}_k + \zeta_k \sBrac{\rBrac{\hat{y}_k - \acute{\mu}_k}^2 - \acute{v}_k},
\end{aligned}
\end{equation*}
where $\gamma_k,~\zeta_k$ satisfy the conditions mentioned above.  The associated normalized estimator is
\begin{equation*}
\frac{(h^2 - 1)}{(h^2 + 1)} \frac{\acute{\mu}_n - \mu^*}{\sqrt{\zeta_n \acute{v}_n}}.
\end{equation*}

\section{Additional numerical experiments}  \label{sec.A.exper}
\setcounter{table}{0} 
\setcounter{figure}{0} 
\setcounter{equation}{0} 
\setcounter{theorem}{0}
\renewcommand{\thetable}{\thesection-\arabic{table}}
\renewcommand{\thefigure}{\thesection-\arabic{figure}}
\renewcommand{\theequation}{\thesection-\arabic{equation}}
\renewcommand{\thetheorem}{\thesection.\arabic{theorem}}

\subsection{One-dimentional cases}
We first consider $\mu(\theta)$ used in \citet{wu2022},
\begin{equation} \label{testfn.1}
    \mu(\theta) = -0.02125 \theta^2 + 0.01825 \theta + 0.0105, ~~~ \theta\in[0,1], 
\end{equation}
and set the same hyperparameters, i.e., $\tau=50,~\theta_0=0.5,~ \mu_0=0,~ v_0=0$, $a_k=30/\rBrac{k+1},~ c_k=1/\rBrac{k+1}^{1/5}$.  Besides, we set $\gamma=0.05$ for our algorithm, $h=3$ for the four-point method and $\gamma_k=1/\rBrac{k+1}^{0.666},~\zeta_k=1/\rBrac{k+1}^{0.55}$ for multi-time-scale method.  Each experiment has been repeated 300 times with $n=1,000,000$ iterations.  

The numerical results (averaged over 300 replications) obtained in the test cases are presented in Table~\ref{table.rmses.plain}. For different algorithms in each case, it presents the sample mean of the root mean square error (RMSE) of the parameter vector and that of the {optimality gap} $\absV{\mu(\theta_n)-\mu^*}$ in sub-tables~\ref{subtab.rmse_plain} and \ref{subtab.diff_plain}, respectively. The associated standard deviations are provided in parentheses.  Additionally, Figure~\ref{fig.plain.rmse} shows the convergence of RMSE with the shadow regions representing 95\% CIs of the Monte Carlo results.  On the performance of optimization, the four-point method shows no significant difference compared to SPSA in terms of RMSE of the parameter vector.  However, forward SPSA results in nearly twice the RMSE of SPSA, even though they share the same order of magnitude.

For statistical inference, Table~\ref{table.tstats.plain} shows the selected means and standard deviations of the normalized estimators, while the associated histograms are depicted in Figure~\ref{fig.plain.histograms}, along with the probability density function of the standard normal distribution (dashed line) as a benchmark to show the validity of normalization.  In Figure~\ref{fig.plain.coverages}, it can be seen that the coverage probabilities of our algorithm, four-method method and forward SPSA would approach the confidence level rapidly in all cases.  

\begin{table}[htbp]  
\centering
\begin{subtable}[H]{0.8\textwidth}
    \centering
    \caption{RMSE of the parameter vector}
    \label{subtab.rmse_plain}
    \csvreader[tabular=cccc, table head= \toprule & \bfseries SPSA$^*$ & \bfseries Forward SPSA & \bfseries Four-point \\ \hline, table foot=\hline]%
    {supplementaries/figures/RMSEs/plain_parameter.csv}{1=\distr, 2=\spsa, 3=\fsp, 4=\fourp}%
    {\distr & \csvcolii & \csvcoliii & \fourp}
\end{subtable}



\scalebox{1}{
\begin{threeparttable}
\begin{subtable}[H]{0.8\textwidth}
    \centering
    \caption{Optimality gap}
    \label{subtab.diff_plain}%
    \csvreader[tabular=cccc, table head= \toprule & \bfseries SPSA & \bfseries Forward SPSA & \bfseries Four-point \\ \hline, table foot=\bottomrule]%
    {supplementaries/figures/RMSEs/plain_performance.csv}{1=\distr, 2=\spsa, 3=\fsp, 4=\fourp}%
    {\distr & \csvcolii & \csvcoliii & \fourp}
\end{subtable}

\begin{tablenotes}
\footnotesize
\item[*] {Our algorithm employs SPSA for optimization, as does the multi-time-scale method.}
\end{tablenotes}
\end{threeparttable}
}

\caption{Performance on test function \eqref{testfn.1} for different cases, based on 300 independent replications.} \label{table.rmses.plain}
\end{table}

\begin{figure}[htbp]  
\centering
\includegraphics[width=\textwidth]{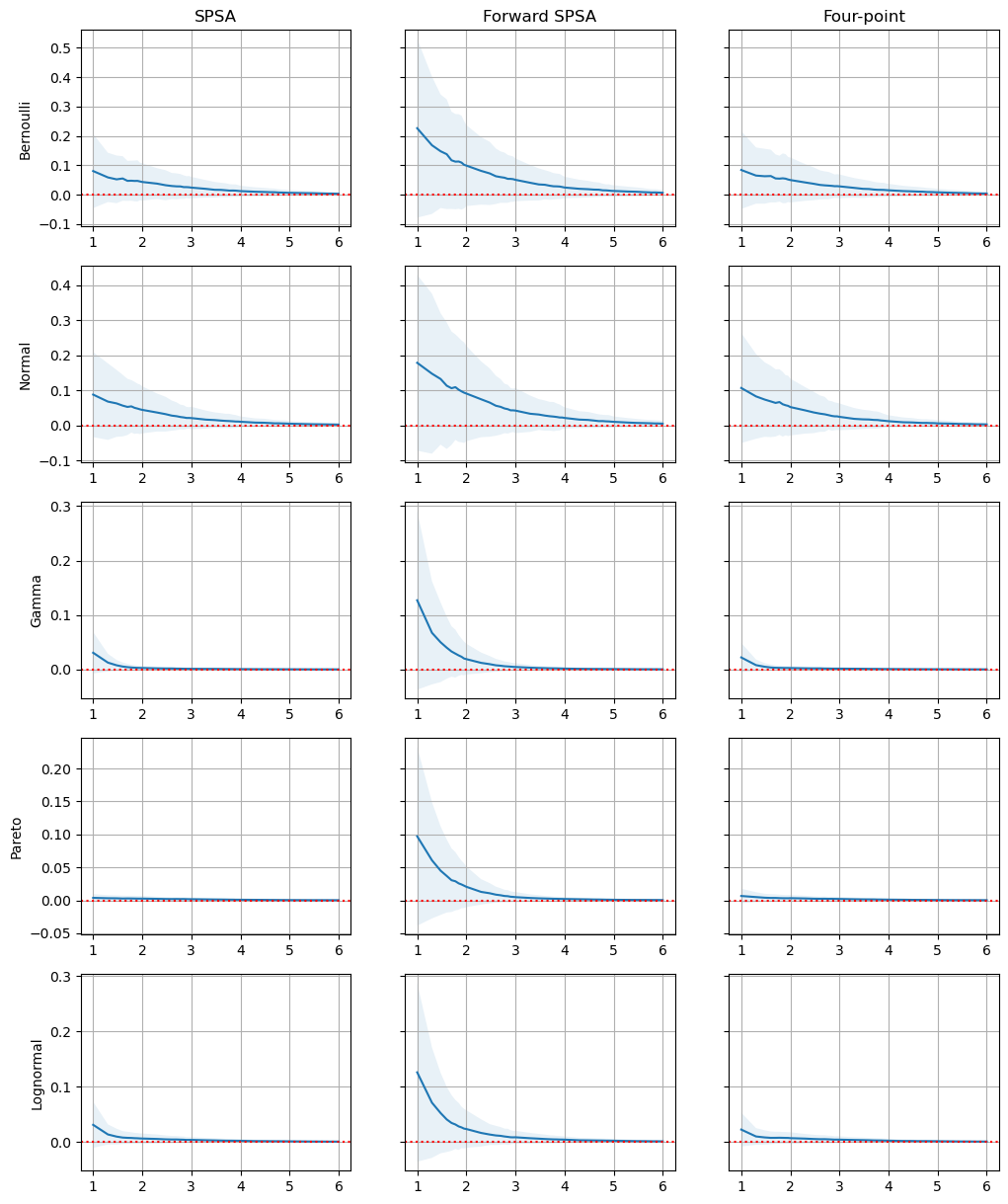}
\caption{RMSEs of input parameters v.s. $\log_{10}\rBrac{\text{\#~iterations}}$ on test function \eqref{testfn.1} for different cases. The shadow regions are 95\% CIs of the Monte Carlo results.}
\label{fig.plain.rmse}
\end{figure}

\begin{table}[htbp]  
\caption{Selected sample means and standard deviations of the normalized estimators on test function \eqref{testfn.1} for different cases, based on 300 independent replications.} \label{table.tstats.plain}
\centering
\begin{filecontents*}{supplementaries/figures/Performances/plain_tStats.csv}
,Ordinary SPSA,Multi-time-scale SPSA,Our algorithm,Forward SPSA,Four-point,Multi-time-scale four-point
Bernoulli,-8.4656 (0.7350),-0.7404 (0.7005),0.0232 (0.9772),-0.2325 (0.9814),-0.1423 (0.9895),-1.3123 (0.7022)
Normal,-10.0575 (0.6495),-0.7993 (0.6714),-0.0828 (1.0428),-0.1554 (1.0177),-0.0834 (1.0492),-1.4767 (0.6716)
Gamma,-139.6910 (0.7090),-11.8404 (0.7084),-0.8123 (0.9041),-0.1641 (0.9517),-0.9507 (0.9546),-21.3710 (0.7128)
Pareto,-121.0359 (1.1307),-10.2523 (0.9060),-0.6803 (1.0449),-0.0072 (0.9905),-0.0153 (0.9948),-18.5360 (1.1732)
Lognormal,-53.2916 (0.6805),-4.5173 (0.7343),-0.2910 (1.0266),-0.0139 (0.9993),-0.3851 (0.9847),-8.1590 (0.7495)
\end{filecontents*}
\csvreader[tabular=ccccc, table head= \toprule & \bfseries Ordinary SPSA & \bfseries Our algorithm & \bfseries Forward SPSA & \bfseries Four-point \\ \hline, table foot=\bottomrule]{supplementaries/figures/Performances/plain_tStats.csv}{}{\csvcoli & \csvcolii & \csvcoliv & \csvcolv & \csvcolvi}
\end{table}

\begin{figure}[htbp]  
\centering
\includegraphics[width=\textwidth]{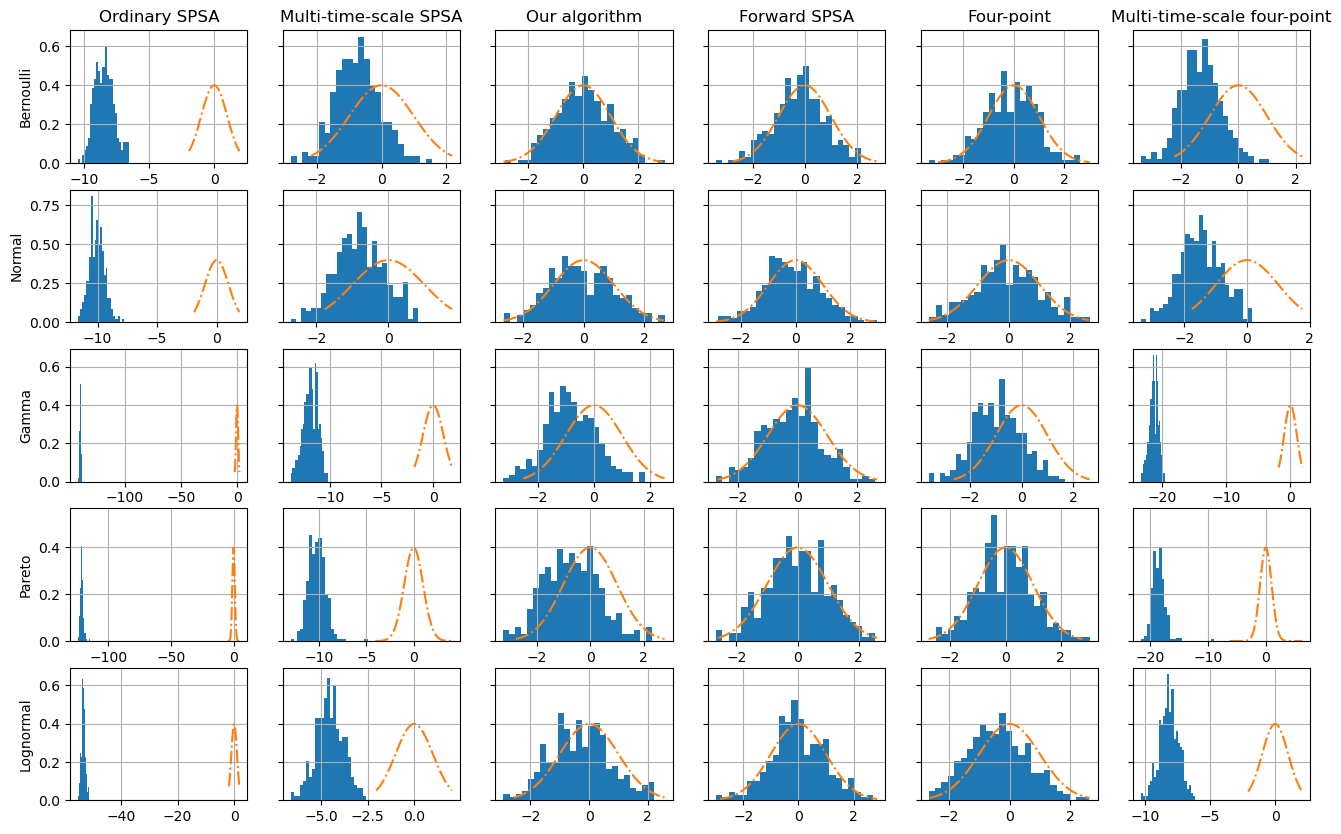}
\caption{Histograms of the normalized estimators on test function \eqref{testfn.1} for different cases. The dashed curve is the probability density function of standard normal distribution.}
\label{fig.plain.histograms}
\end{figure}

\begin{figure}[htbp]  
\centering
\includegraphics[width=\textwidth]{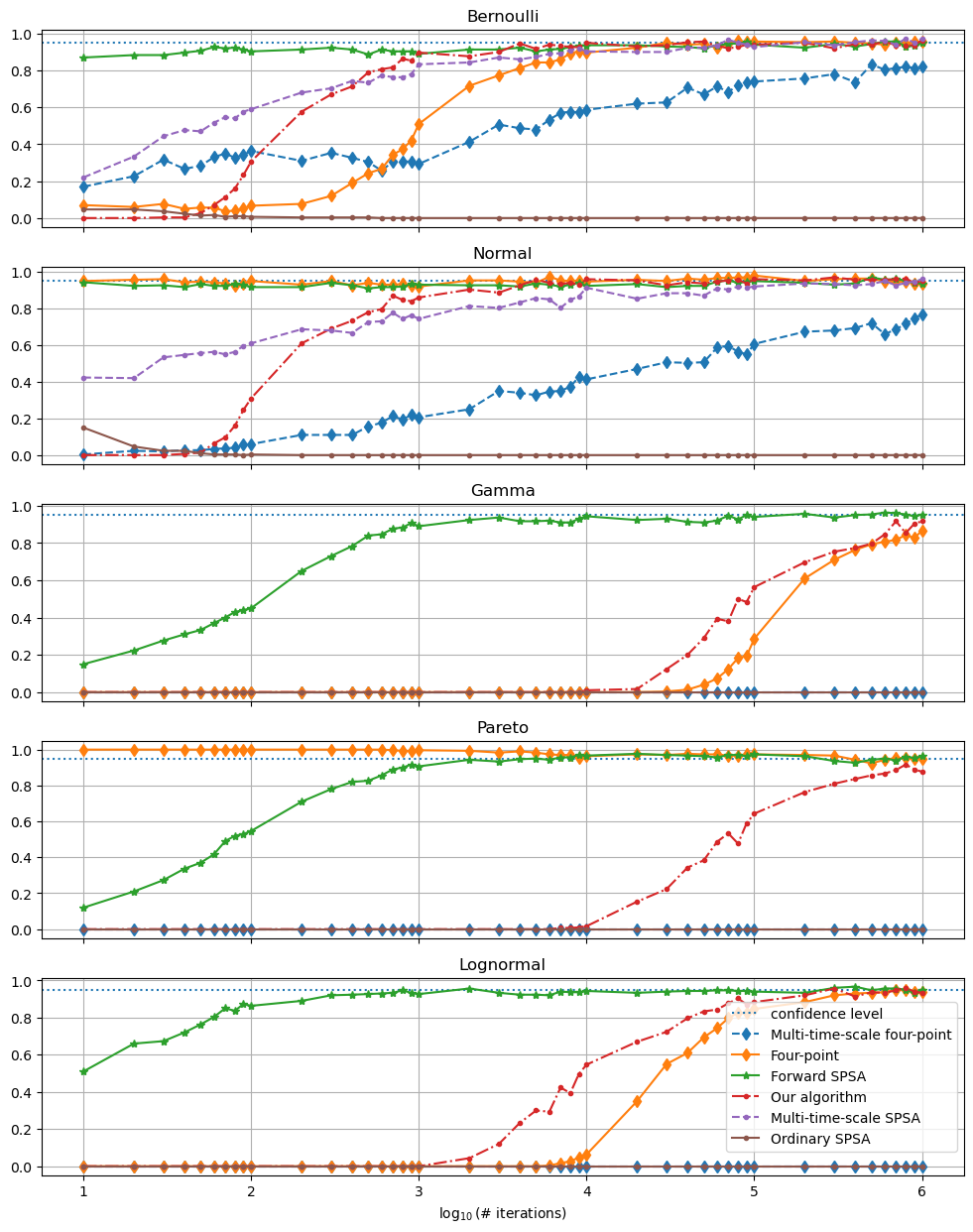}
\caption{Coverage probabilities of the asymptotic CIs on test function \eqref{testfn.1} for different cases. The horizontal dashed line with no markers indicates the confidence level.}
\label{fig.plain.coverages}
\end{figure}

Now we consider the test function
\begin{equation} \label{testfn.2}
\mu(\theta) = \theta^2 - 2\theta + 1.5, ~~~ \theta\in[-2,2],
\end{equation}
which has a wider range than \eqref{testfn.1}.
Except for setting $\ex{Y(\theta)} = \sBrac{1 + \exp\cBrac{3-\mu(\theta)}}\inv$ in \ref{case.bernoulli} and $n=100,000$ in call cases, the other hyperparameters take the same values as above.  For optimization, because the results are similar and it is not the point of this paper, we omit the numerical outputs for the below cases.
Similar results for the convergence of parameter vectors are presented in Table~\ref{table.rmses.sharp} and Figure~\ref{fig.sharp.rmse}, respectively.

For the test function \eqref{testfn.2}, our algorithm outperforms the other methods in terms of statistical inference.  In Figure~\ref{fig.sharp.histograms} the histograms of the normalized estimator generated by our algorithm are closer to the standard normal distribution in all cases.  Additionally, the sample means and standard deviations of the normalized estimators of ordinary SPSA and those with reduced bias are presented in Table~\ref{table.tstats.sharp}.  Meanwhile, as shown in Figure~\ref{fig.sharp.coverages}, it can be seen that the coverage probabilities of our algorithm approach the confidence level more rapidly than the others.  It is challenging to ensure that all the aforementioned algorithms yield satisfactory results within a finite sample under a specified set of hyperparameters. Notably, even though the empirical coverage probabilities of all algorithms deviating significantly from the confidence level in \ref{case.gamma} and \ref{case.pareto}, our algorithm results in a smaller error.

\begin{table}[htbp]  
\centering

\begin{subtable}[H]{0.8\textwidth}
    \centering
    \caption{RMSE of the parameter vector}
    \label{subtab.rmse_sharp}
    \csvreader[tabular=cccc, table head= \toprule & \bfseries SPSA$^*$ & \bfseries Forward SPSA & \bfseries Four-point \\ \hline, table foot=\hline]%
    {supplementaries/figures/RMSEs/sharp_parameter.csv}{1=\distr, 2=\spsa, 3=\fsp, 4=\fourp}%
    {\distr & \csvcolii & \csvcoliii & \fourp}
\end{subtable}



\scalebox{1}{
\begin{threeparttable}
\begin{subtable}[H]{0.8\textwidth}
    \centering
    \caption{Optimality gap}
    \label{subtab.diff_sharp}%
    \csvreader[tabular=cccc, table head= \toprule & \bfseries SPSA & \bfseries Forward SPSA & \bfseries Four-point \\ \hline, table foot=\bottomrule]%
    {supplementaries/figures/RMSEs/sharp_performance.csv}{1=\distr, 2=\spsa, 3=\fsp, 4=\fourp}%
    {\distr & \csvcolii & \csvcoliii & \fourp}
\end{subtable}

\begin{tablenotes}
\footnotesize
\item[*] {Our algorithm employs SPSA for optimization, as does the multi-time-scale method.}
\end{tablenotes}
\end{threeparttable}
}

\caption{Performance on test function \eqref{testfn.2} for different cases, based on 300 independent replications.} \label{table.rmses.sharp}
\end{table}

\begin{figure}[htbp]  
\centering
\includegraphics[width=\textwidth]{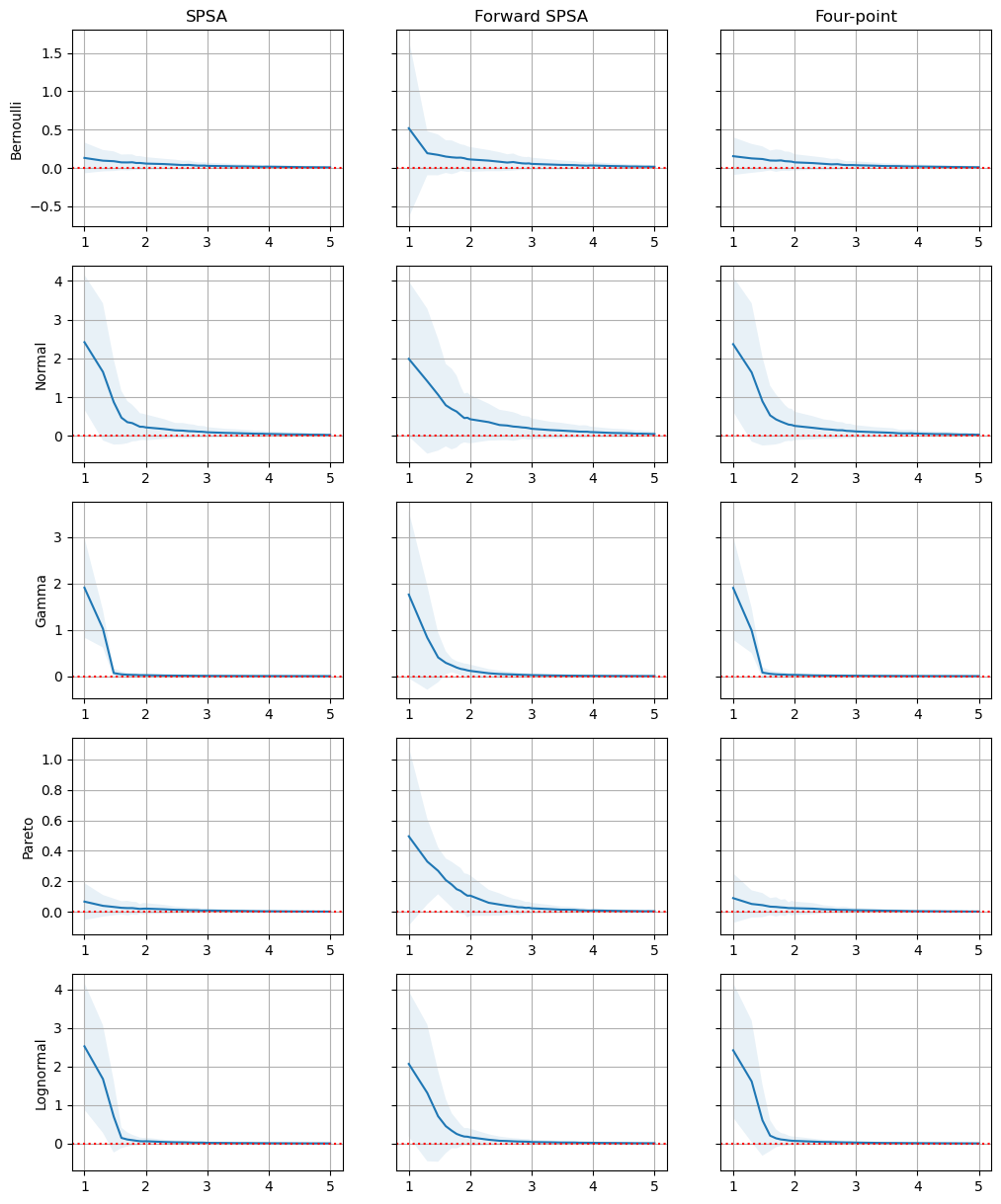}
\caption{RMSEs of input parameters v.s. $\log_{10}\rBrac{\text{\#~iterations}}$ on test function \eqref{testfn.2} for different cases. The shadow regions are 95\% CIs of the Monte Carlo results.}
\label{fig.sharp.rmse}
\end{figure}

\begin{table}[htbp]  
\caption{Selected sample means and standard deviations of the normalized estimators on the test function \eqref{testfn.2}, based on 300 independent replications.} \label{table.tstats.sharp}
\centering
\begin{filecontents*}{supplementaries/figures/Performances/sharp_tStats.csv}
,Ordinary SPSA,Multi-time-scale SPSA,Our algorithm,Forward SPSA,Four-point,Multi-time-scale four-point
Bernoulli,10.2472 (0.6658),1.1864 (0.7215),0.1361 (0.9286),0.7325 (0.9250),-0.7307 (0.9378),2.2061 (0.7151)
Normal,13.0953 (0.6552),1.4956 (0.7368),0.2530 (1.0198),6.2975 (1.0701),3.0280 (0.9156),2.6514 (0.7444)
Gamma,142.5703 (0.8382),18.3325 (0.7133),2.3788 (0.9372),5.9833 (1.2491),5.8352 (0.9994),32.2802 (0.6908)
Pareto,125.0041 (2.5170),15.9508 (1.0965),2.1765 (1.1388),1.9140 (1.0204),0.7285 (0.9637),28.0861 (1.7864)
Lognormal,57.0046 (0.7592),6.9424 (0.6537),0.8974 (1.0137),4.7949 (1.1223),4.0915 (1.0122),12.3266 (0.6268)
\end{filecontents*}
\csvreader[tabular=ccccc, table head= \toprule & \bfseries Ordinary SPSA & \bfseries Our algorithm & \bfseries Forward SPSA & \bfseries Four-point \\ \hline, table foot=\bottomrule]{supplementaries/figures/Performances/sharp_tStats.csv}{}{\csvcoli & \csvcolii & \csvcoliv & \csvcolv & \csvcolvi}
\end{table}

\begin{figure}[htbp]  
\centering
\includegraphics[width=\textwidth]{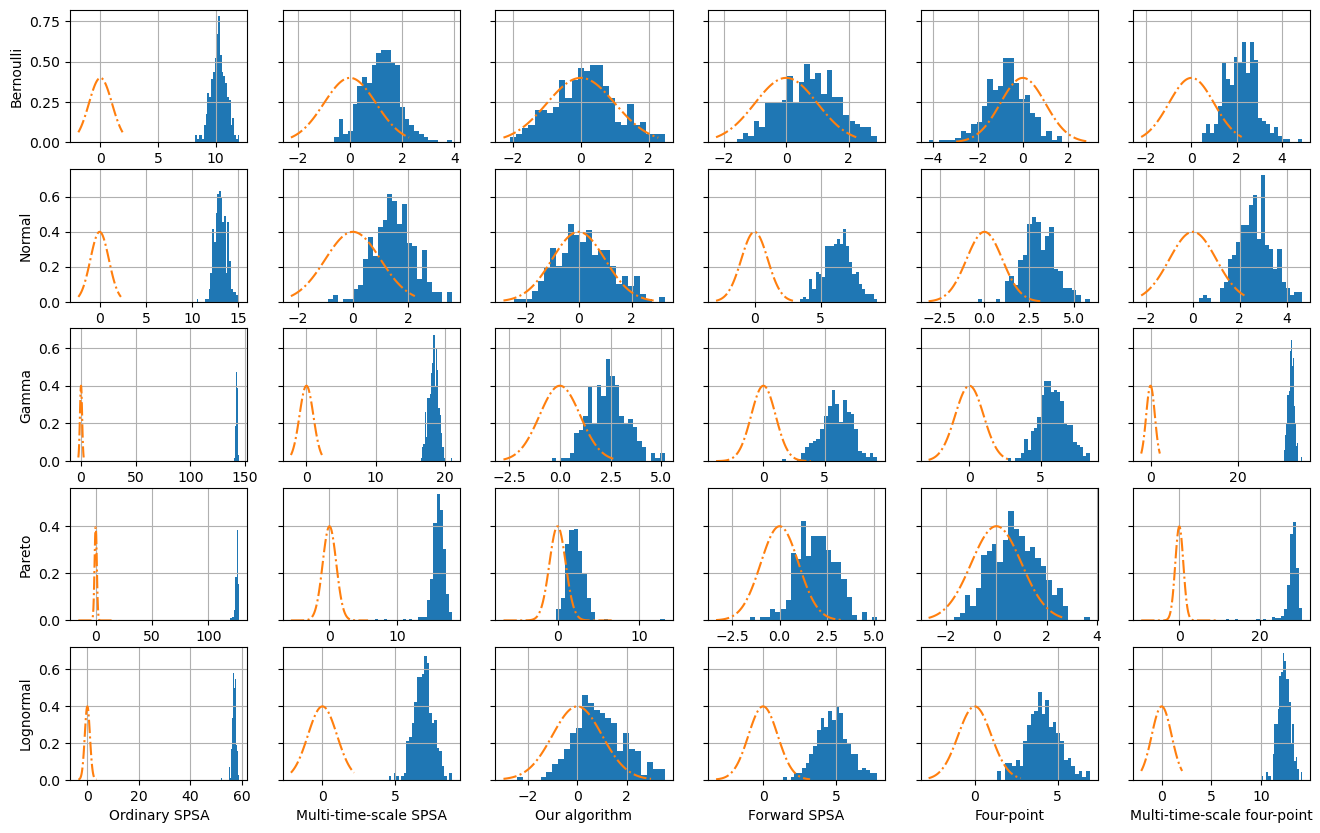}
\caption{Histograms of the normalized estimators on the test function \eqref{testfn.2}. The dashed curve is the probability density function of standard normal distribution.}
\label{fig.sharp.histograms}
\end{figure}

\begin{figure}[htbp]  
\centering
\includegraphics[width=\textwidth]{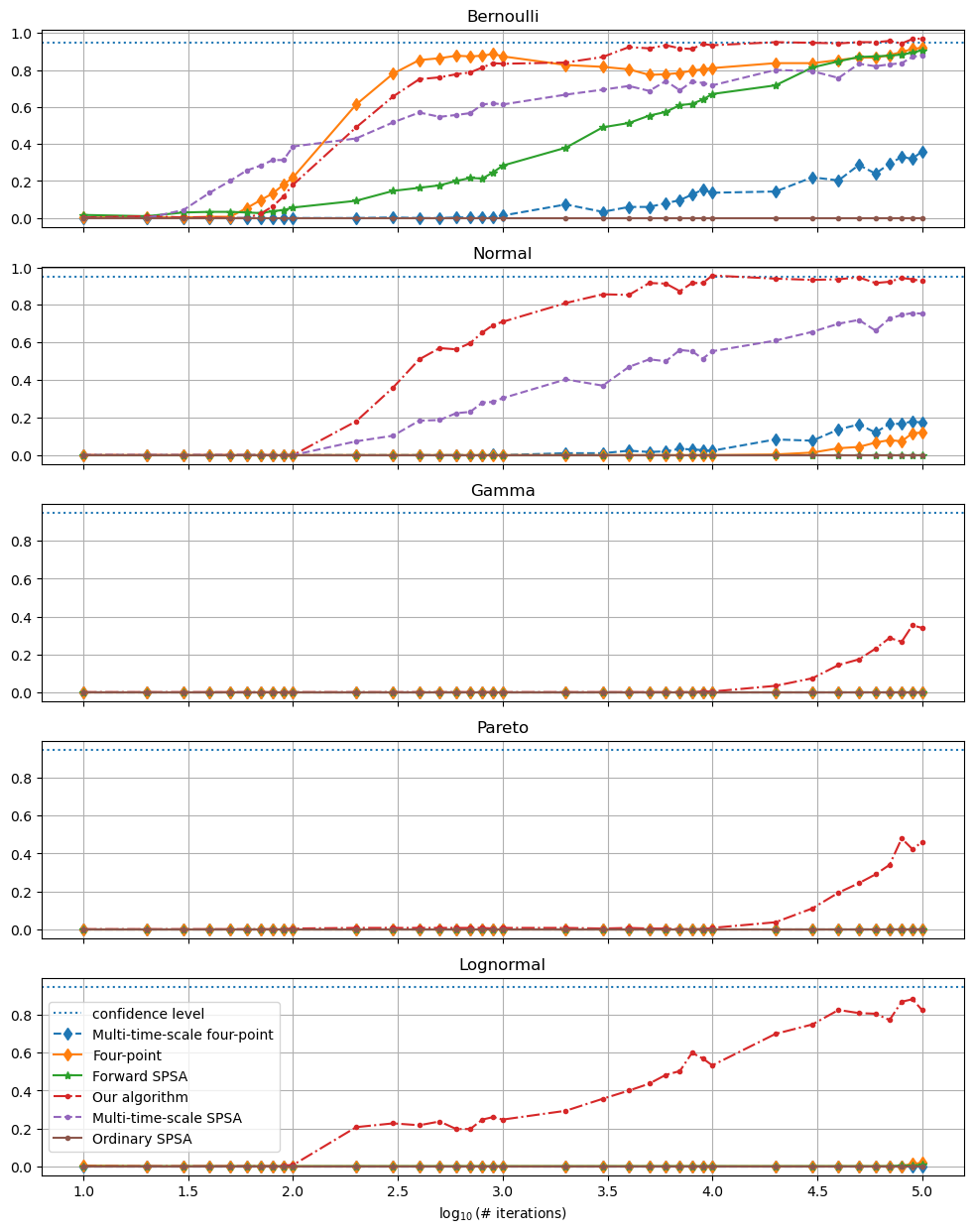}
\caption{Coverage probabilities of the asymptotic CIs on the test function \eqref{testfn.2}. The horizontal dashed line with no markers indicates the confidence level.}
\label{fig.sharp.coverages}
\end{figure}

As a matter of fact, the results above offer numerical evidence supporting the conclusion of Section~\ref{sec.compr}.  Typically, exogenous errors are dominated by the endogenous errors, and the endogenous bias of the normalized estimator generated by our algorithm converges to 0 more rapidly than that of the benchmarks.  Our algorithm demonstrates a superior capacity to reduce the impact of endogenous errors, whereas the benchmarks either fail to mitigate it or require larger sample size.

\subsection{Multi-dimentional case}
Consider the Perm (0,10,10) function, which can be found at \url{https://www.sfu.ca/~ssurjano/perm0db.html},
\begin{equation*}
f(\theta) = \sum_{i=1}^{10} \sBrac{\sum_{j=1}^{10}(j+10)\rBrac{\theta_j^i - j^{-i}}}^2,~~~ \theta \in [-2,2]^{10}.
\end{equation*}
We test in \ref{case.bernoulli} and \ref{case.normal} with $\mu(\theta) = 1/\sBrac{1+ \exp\cBrac{-A_1 f(\theta)-C_1}}$ and $\mu(\theta) = A_2 f(\theta)+C_2$, respectively, where $A_1 = 2\times 10^{-11},~ C_1=-1$ and $A_2=10^{-10}, C_2=0$.  We set $\tau=20,~ a_k=1000/\rBrac{k+1000},~ c_k=1/\rBrac{k+1}^{1/6},~ n=1,000,000$, and initialize $\theta_0$ uniformly distributed in $[-2,2]^{10}$.

We repeat 300 times and report the sample means and standard deviations of the normalized estimators of selected algorithms in Table~\ref{table.tstats.TenDimPerm}, and the corresponding histograms in Figure~\ref{fig.TenDimPerm.histograms}.  At last, the coverage probabilities of the asymptotic 95\%-CIs are depicted in the Figure~\ref{fig.TenDimPerm.coverages}.  

\begin{table}[htbp]  
\caption{Selected sample means and standard deviations of the normalized estimators on the Perm (0,10,10) function, based on 300 independent replications.} \label{table.tstats.TenDimPerm}
\centering
\begin{filecontents*}{supplementaries/figures/Performances/TenDimPerm_tStats.csv}
,Ordinary SPSA,Multi-time-scale SPSA,Our algorithm,Forward SPSA,Four-point,Multi-time-scale four-point
Bernoulli,22.1932 (1.0306),1.6460 (1.0219),0.1144 (0.9641),19.3518 (1.0966),14.7937 (7.4331),2.2520 (2.7113)
Normal,301.2336 (5.0426),9.0646 (1.2396),0.4290 (0.9239),285.0424 (1.2257),7.6162 (5.7011),8.7339 (5.9289)
\end{filecontents*}
\csvreader[tabular=ccccc, table head= \toprule & \bfseries Ordinary SPSA & \bfseries Our algorithm & \bfseries Forward SPSA & \bfseries Four-point \\ \hline, table foot=\bottomrule]{supplementaries/figures/Performances/TenDimPerm_tStats.csv}{}{\csvcoli & \csvcolii & \csvcoliv & \csvcolv & \csvcolvi}
\end{table}

\begin{figure}[htbp]  
\centering
\includegraphics[width=\textwidth]{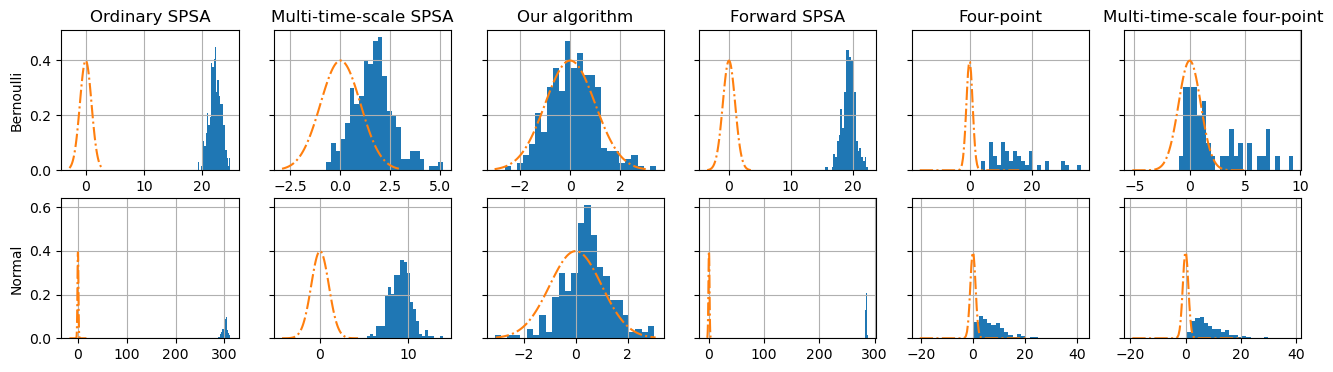}
\caption{Histograms of the normalized estimators on the Perm (0,10,10) function. The dashed curve is the probability density function of standard normal distribution.}
\label{fig.TenDimPerm.histograms}
\end{figure}

\begin{figure}[htbp]  
\centering
\includegraphics[width=\textwidth]{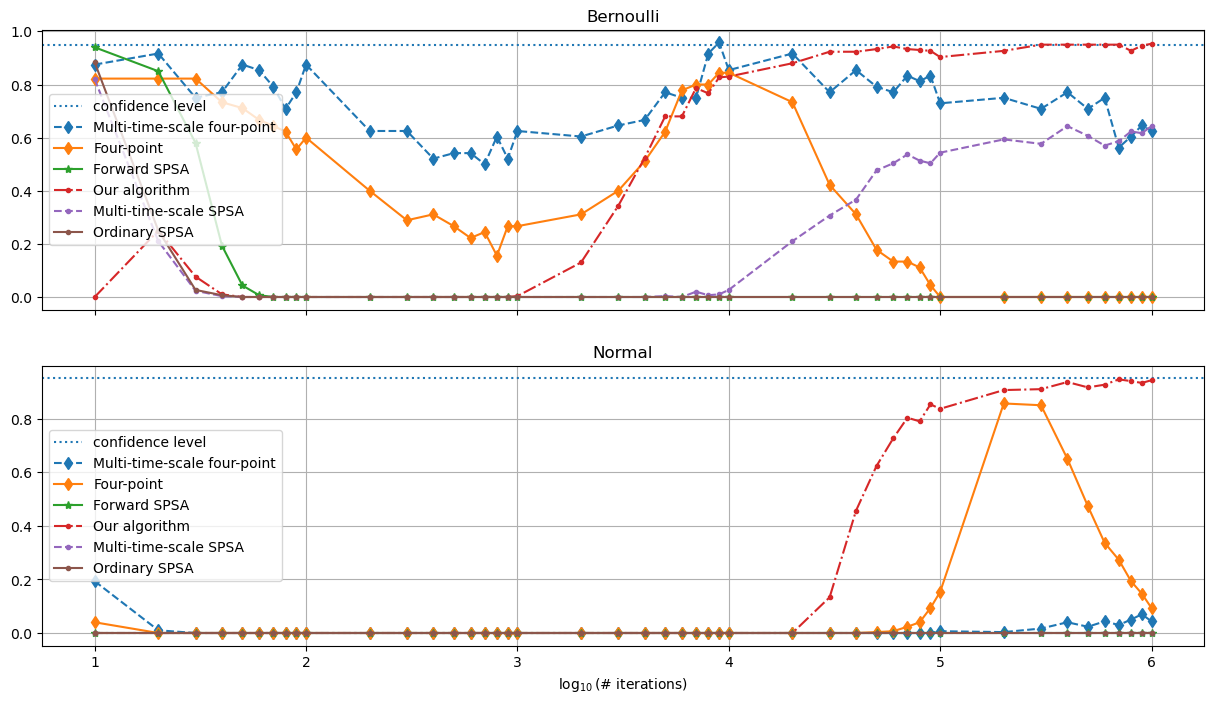}
\caption{Coverage probabilities of the asymptotic CIs on the Perm (0,10,10) function. The horizontal dashed line with no markers indicates the confidence level.}
\label{fig.TenDimPerm.coverages}
\end{figure}

\end{appendices}

\end{document}